\documentclass{amsart}

\usepackage[latin1]{inputenc}
\usepackage{indentfirst}
\usepackage[dvips]{graphicx}
\usepackage{graphicx}
\usepackage{newlfont}
\usepackage{amssymb}
\usepackage{amsmath}
\usepackage{latexsym}
\usepackage{amsthm}
\usepackage{psfrag}
\usepackage[usenames]{color}

\newtheorem{theorem}{Theorem}[section]
\newtheorem{lemma}[theorem]{Lemma}
\newtheorem{cor}[theorem]{Corollary}
\newtheorem{prop}[theorem]{Proposition}

\theoremstyle{definition}
\newtheorem{definition}[theorem]{Definition}
\newtheorem{example}[theorem]{Example}

\theoremstyle{remark}
\newtheorem{remark}[theorem]{Remark}

\newcommand{\R}{\mathbb R}

\newcommand{\x}{\vec x}
\newcommand{\y}{\vec y}

\newcommand{\fr}{\vec \varphi}

\newcommand{\M}{{\mathcal M}}

\numberwithin{equation}{section}

\begin{document}

\title[Discontinuities of multidimensional size functions]{Necessary conditions for discontinuities of multidimensional size functions}

\author{A. Cerri}
\address{Andrea Cerri, ARCES, Universit\`a di Bologna,
via Toffano $2/2$, I-$40135$ Bologna, Italia\newline Dipartimento di Matematica, Universit\`a di Bologna,
P.zza di Porta S. Donato 5, I-$40126$ Bologna,
Italia}
\email{cerri@dm.unibo.it}

\author{P. Frosini}
\address{Patrizio Frosini (corresponding author), ARCES, Universit\`a di Bologna,
via Toffano $2/2$, I-$40135$ Bologna, Italia\newline Dipartimento di Matematica, Universit\`a di Bologna,
P.zza di Porta S. Donato 5, I-$40126$ Bologna,
Italia, tel. +39-051-2094478, fax. +39-051-2094490}
\email{frosini@dm.unibo.it}


\date{\today} 



\begin{abstract}
Some new results about multidimensional Topological Persistence
are presented, proving that the discontinuity points of a
$k$-dimensional size function are necessarily related to the
pseudocritical or special values of the associated measuring
function.\end{abstract}

\maketitle

\section*{Introduction}
Topological Persistence is devoted to the study of stable
properties of sublevel sets of topological spaces and, in the course of its development, has revealed itself to be
a suitable framework when dealing with applications in the field
of Shape Analysis and Comparison.
Since the beginning of the 1990s research on this subject has
been carried out under the name of Size Theory, studying the
concept of {\em size function}, a mathematical tool able to
describe the qualitative properties of a shape in a quantitative
way. More precisely, the main idea is to model a shape by a
topological space $\M$ endowed with a continuous function
$\varphi$, called {\em measuring function}. Such a function is
chosen according to applications and can be seen as a descriptor
of the features considered relevant for shape characterization.
Under these assumptions, the size function $\ell_{(\M,\varphi)}$
associated with the pair $(\M,\varphi)$ is a descriptor of the
topological attributes that persist in the sublevel sets of $\M$
induced by the variation of $\varphi$. According to this approach,
the problem of comparing two shapes can be reduced to the simpler
comparison of the related size functions. Since their
introduction, these shape descriptors have been widely studied and
applied in quite a lot of concrete applications concerning Shape
Comparison and Pattern Recognition (cf., e.g.,
\cite{BiGiSpFa08,CeFeGi06,DiFrPa04,UrVe97,VeUr96,VeUrFrFe93}).
From a more theoretical point of view, the notion of size function
plays an essential role since it is strongly related  to the one
of {\em natural pseudodistance}. This is another key tool of Size
Theory, defining a (dis)similarity measure between compact and
locally connected topological spaces endowed with measuring
functions (see \cite{BiDeFaFrGiLaPaSp08} for historical references
and \cite{DoFr04,DoFr07,DoFr} for a detailed review about the
concept of natural pseudodistance). Indeed, size functions provide
easily computable lower bounds for the natural pseudodistance (cf.
\cite{dAFrLabis,dAFrLa,DoFr04bis}).

Approximately ten years after the introduction of Size Theory,
Persistent Homology re-proposed similar ideas from the homological
point of view (cf. \cite{EdLeZo02}; for a survey on this topic see
\cite{EdHa08}). In this context, the notion of size function
coincides with the dimension of the $0$-th persistent homology
group, i.e. the $0$-th rank invariant \cite{CaZo09}.

We refer the interested reader to Appendix \ref{Appendix} for more information
about the relationship existing between Size Theory and Persistent
Homology.

The study of Topological Persistence is capturing more and more
attention in the mathematical community, with particular reference
to the multidimensional setting (see \cite{EdHa08,Gh08}). When
dealing with size functions, the term {\em multidimensional} means
that the measuring functions are vector-valued, and has no
reference to the dimension of the topological space under study.
However, while the basic properties of a size function $\ell$ are
now clear when it is associated with a measuring function $\varphi$
taking values in $\mathbb{R}$, very little is known when $\varphi$
takes values in $\R^k$. More precisely, some questions about the
structure of size functions associated with $\R^k$-valued
measuring functions need further investigation, with particular
reference to the localization of their discontinuities. Indeed,
this last research line is essential in the development of
efficient algorithms allowing us to apply Topological Persistence
to concrete problems in the multidimensional context.

In this paper we start to fill this gap by proving a new result on
the discontinuities of the so-called multidimensional size
functions, showing that they can be located only at points with at
least one \emph{pseudocritical} or \emph{special} coordinate
(Theorem \ref{Teoremone} and Theorem \ref{Teoremone2}). This is
proved by using an approximation technique and the theoretical
machinery developed in \cite{BiCeFrGiLa08}, improving the
comprehension of multidimensional Topological Persistence and
laying the basis for its computation.

This paper is organized in two sections. In Section
\ref{BasicResult} the basic results about multidimensional size
functions are recalled, while in Section 2 our main theorems are
proved.

\section{Preliminary Results on Size Theory}\label{BasicResult}

The main idea in Size Theory is to study a given shape by
performing a geometrical/topological exploration of a suitable
topological space $\mathcal{M}$, with respect to some properties
expressed by an $\mathbb{R}^k$-valued continuous function
$\vec\varphi=(\varphi_1,\dots,\varphi_k)$ defined on
$\mathcal{M}$. Following this approach, Size Theory introduces the
concept of {\em size function} as a stable and compact descriptor
of the topological changes occurring in the lower level sets
$\{P\in\mathcal{M}:\varphi_i(P)\leq t_i,i=1,\dots,k\}$ as $\vec
t=(t_1,\dots,t_k)$ varies in $\mathbb{R}^k$.

In this section we recall some basic definitions and results about
size functions, confining ourselves to those that will be useful
in the rest of this paper. For a deeper investigation on these
topics, the reader is referred to
\cite{BiCeFrGiLa08,BiDeFaFrGiLaPaSp08,FrMu99}. For further details
about Topological Persistence in the multidimensional setting, see
\cite{CaZo09,FrMu99}.

In proving our new results we need to assume that $\M$ is a closed
$C^1$ Riemannian manifold. However, we prefer to report here the
basic concepts of Size Theory in their classical formulation, i.e.
by supposing that $\mathcal{M}$ is a non-empty compact and locally
connected Hausdorff space. We shall come back to the case of a
$C^1$ Riemannian manifold later.

In the context of Size Theory, any pair $(\mathcal{M},\fr)$, where
$\fr=(\varphi_1,\dots,\varphi_k):\mathcal{M}\rightarrow\R^k$ is a
continuous function, is called a \emph{size pair}. The function
$\fr$ is said to be a \emph{$k$-dimensional measuring function}.
The relations $\preceq$ and $\prec$ are defined in $\R^k$ as
follows: for $\x=(x_1,\dots,x_k)$ and $\y=(y_1,\dots,y_k)$, we
write $\x\preceq\y$ (resp. $\x\prec\y$) if and only if $x_i\leq\
y_i$ (resp. $x_i<y_i$)
for every index $i=1,\dots,k$. Furthermore, 
$\R^k$ is equipped 
with the usual $\max$-norm:
$\left\|(x_1,x_2,\dots,x_k)\right\|_{\infty}=\max_{1\le i\le
k}|x_i|$. Now we are ready to introduce the concept of size
function for a size pair $(\mathcal{M},\fr)$. We shall denote the
open set $\{(\x,\y)\in\R^k\times\R^k:\x\prec \y\}$ by $\Delta^+$,
while $\bar\Delta^+$ will be the closure of $\Delta^+$. For every
$k$-tuple $\x=(x_1,\dots,x_k)\in\R^k$, the set
$\mathcal{M}\langle\fr\preceq \x\,\rangle$  will be defined as
$\{P\in\mathcal{M}:\varphi_i(P)\leq x_i,\ i=1,\dots,k\}$.

\begin{definition}\label{firstDef}
For every $k$-tuple $\y=(y_1,\dots,y_k)\in\R^k$, we shall say that
two points $P,Q\in \mathcal{M}$ are $\langle\fr\preceq
\y\,\rangle$-\emph{connected} if and only if a connected subset of
$\mathcal{M}\langle\fr\preceq \y\,\rangle$ exists, containing $P$
and $Q$.
\end{definition}

\begin{definition}\label{secondDef}
We shall call the \emph{($k$-dimensional) size function}
associated with the size pair $(\mathcal{M},\fr)$ the function
$\ell_{(\mathcal{M},\fr)}:\Delta^+\rightarrow\mathbb{N}$, defined
by setting $\ell_{(\mathcal{M},\fr)}(\x,\y)$ equal to the number
of equivalence classes in which the set
$\mathcal{M}\langle\fr\preceq \x\,\rangle$ is divided by the
$\langle\fr\preceq \y\,\rangle$-connectedness relation.
\end{definition}

\begin{remark}\label{altdef}
In other words, $\ell_{(\mathcal{M},\fr)}(\x,\y)$ is equal to the
number of connected components in $\mathcal{M}\langle\fr\preceq
\y\,\rangle$ containing at least one point of
$\mathcal{M}\langle\fr\preceq \x\,\rangle$. The finiteness of this
number is a consequence of the compactness and local connectedness of $\M$ (cf. \cite{FrLa97}).
\end{remark}

In the following, we shall refer to the case of measuring
functions taking value in $\mathbb{R}^k$ by using the term
``$k$-dimensional''. Before going on, we introduce the following
notations: when $\vec y\in\R^k$ is fixed, the symbol
$\ell_{(\M,\vec\varphi)}(\cdot,\vec y)$ will be used to denote the
function that takes each $k$-tuple $\vec x\prec\vec y$ to the
value $\ell_{(\M,\vec\varphi)}(\vec x,\vec y)$. An analogous
convention will hold for the symbol $\ell_{(\M,\vec\varphi)}(\vec
x,\cdot)$.

\begin{remark}\label{monotone}
From Remark \ref{altdef} it can be immediately deduced that for
every fixed $\vec y\in\R^k$ the function
$\ell_{(\M,\varphi)}(\cdot,\vec y)$ is non--decreasing with
respect to $\preceq$, while for every fixed $\vec x\in\R^k$ the
function $\ell_{(\M,\varphi)}(\vec x,\cdot)$ is non--increasing.
\end{remark}

\subsection{The particular case {\boldmath $k=1$}}\label{PartCase1}

In this section we will discuss the specific framework of
measuring functions taking values in $\mathbb{R}$, namely the
$1$-dimensional case. Indeed, Size Theory has been extensively
developed in this setting (cf. \cite{BiDeFaFrGiLaPaSp08}), showing
that each $1$-dimensional size function admits a compact
representation as a formal series of points and lines of
$\mathbb{R}^2$ (cf. \cite{FrLa01}). Due to this representation, a
suitable {\em matching distance} between $1$-dimensional size
functions can be easily introduced, proving that these descriptors
are stable with respect to such a distance \cite{dA02,dAFrLa}.
Moreover, the role of $1$-dimensional size functions is crucial in
the approach to the $k$-dimensional case proposed in
\cite{BiCeFrGiLa08}.

Following the notations used in the literature about the case
$k=1$, the symbols $\vec\varphi$, $\vec x$, $\vec y$, $\preceq$, $\prec$ will be
replaced respectively by $\varphi$, $x$, $y$, $\le$, $<$.

When dealing with a ($1$-dimensional) measuring function
$\varphi:\M\to\R$, the size function $\ell_{(\M,\varphi)}$
associated with $(\M,\varphi)$ gives information about the pairs
$\left(\mathcal{M}\langle\varphi\leq
x\rangle,\mathcal{M}\langle\varphi\leq y\rangle\right)$, where
$\mathcal{M}\langle\varphi\leq t\rangle$ is defined by setting
$\mathcal{M}\langle\varphi\leq t\rangle=\{P\in\M:\varphi(P)\leq
t\}$ for $t\in\R$.

Figure \ref{FigEs} shows an example of a size pair and the
associated $1$-dimensional size function.
\begin{figure}[h]
\psfrag{M}{{\Large $\mathcal{M}$}}
\psfrag{F}{{\large $\varphi$}}
\psfrag{N}{\!\!\!{\Large $\ell_{(\mathcal{M},\varphi)}$}}
\psfrag{(A)}{{\large $(a)$}}\psfrag{(E)}{$(b)$}
\psfrag{x}{{\large $x$}}\psfrag{y}{\large $y$}
\psfrag{a}{\textcolor[rgb]{0.00,0.00,1.00}{$a$}}
\psfrag{b}{\textcolor[rgb]{0.00,0.00,1.00}{$b$}}
\begin{center}
\includegraphics[width=0.8\textwidth]{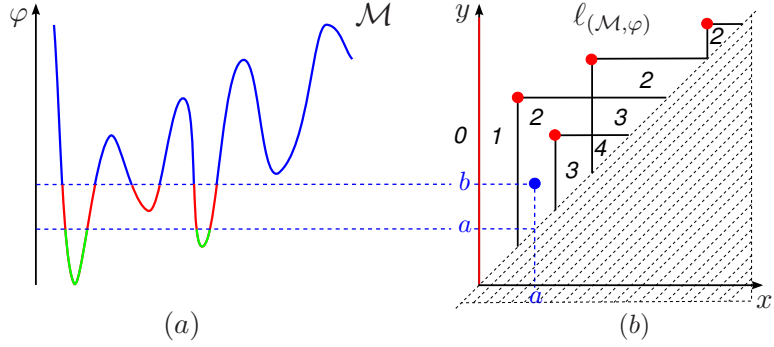}
\end{center}
\caption{$(a)$ The topological spaces $\mathcal{M}$ and the
measuring function $\varphi$. $(b)$ The related size function
$\ell_{(\M,\varphi)}$.}\label{FigEs}
\end{figure}
On the left (Figure \ref{FigEs}$(a)$) one can find the considered
size pair $(\M,\varphi)$, where $\M$ is the curve depicted by a
solid line, and $\varphi$ is the ordinate function. On the right
(Figure \ref{FigEs}$(b)$) the associated $1$-dimensional size
function $\ell_{(\M,\varphi)}$ is given. As can be seen, the
domain $\Delta^+=\{(x,y)\in\R^2:x<y\}$ is divided into bounded and
unbounded regions, in each of which the $1$-dimensional size
function takes a constant value. The displayed numbers coincide
with the values of $\ell_{(\M,\varphi)}$ in each region. For
example, let us now compute the value of $\ell_{(\M,\varphi)}$ at
the point $(a,b)$. By applying Remark \ref{altdef} in the case
$k=1$, it is sufficient to count how many of the three connected
components in the sublevel $\mathcal{M}\langle\varphi\leq
b\rangle$ contain at least one point of
$\mathcal{M}\langle\varphi\leq a\rangle$. It can be easily
verified that $\ell_{(\M,\varphi)}(a,b)=2$.

Following the $1$-dimensional framework, the problem of comparing
two size pairs can be easily translated into the simpler one of
comparing the related $1$-dimensional size functions. In
\cite{dAFrLa}, the {\em matching distance} $d_{match}$ has been formally proven
to be the most suitable distance between these descriptors. The
definition of $d_{match}$ is based on the observation that
$1$-dimensional size functions can be compactly described by a
formal series of points and lines lying on the real plane, called
respectively {\em proper cornerpoints} and {\em cornerpoints at
infinity} (or {\em cornerlines}) and defined as follows:

\begin{definition}
For every point $P=(x,y)$ with $x<y$, consider the number $\mu(P)$
defined as the minimum, over all the positive real numbers
$\varepsilon$ with $x+\varepsilon<y-\varepsilon$, of
$$
\ell_{(\M,\varphi)}(x+\varepsilon,y-\varepsilon)-\ell_{(\M,\varphi)}(x-\varepsilon,y-\varepsilon)
-\ell_{(\M,\varphi)}(x+\varepsilon,y+\varepsilon)+\ell_{(\M,\varphi)}(x-\varepsilon,y+\varepsilon).
$$
When this finite number, called {\em multiplicity of $P$}, is
strictly positive, the point $P$ will be called a {\em proper
cornerpoint} for $\ell_{(\M,\varphi)}$.
\end{definition}

\begin{definition}
For every line $r$ with equation $x=a$, consider the number
$\mu(r)$ defined as the minimum, over all the positive real
numbers $\varepsilon$ with $a+\varepsilon<1/\varepsilon$, of
$$
\ell_{(\M,\varphi)}(a+\varepsilon,1/\varepsilon)-\ell_{(\M,\varphi)}(a-\varepsilon,1/\varepsilon).
$$
When this finite number, called {\em multiplicity of $r$}, is
strictly positive, the line $r$ will be called a {\em cornerpoint
at infinity} (or {\em cornerline})  for $\ell_{(\M,\varphi)}$.
\end{definition}

The fundamental role of proper cornerpoints and cornerpoints at
infinity is explicitly shown in the following Representation
Theorem, claiming that their multiplicities completely and
univocally determine the values of $1$-dimensional size functions.

For the sake of simplicity, each line of equation $x=a$ will be
identified to a point at infinity with coordinates $(a,\infty)$.

\begin{theorem}[Representation Theorem]\label{RepresentationTheorem}
For every $\bar x<\bar y<\infty$, it holds that
$$
\ell_{(\M,\varphi)}(\bar x,\bar y)=\sum_{
x\leq\bar x \atop
\bar y<y\leq\infty}\mu((x,y)).
$$
\end{theorem}
\begin{remark}
In plain words, the Representation Theorem
\ref{RepresentationTheorem} claims that the value
$\ell_{(\M,\varphi)}(\bar x,\bar y)$ equals the number of
cornerpoints lying above and on the left of $(\bar x,\bar y )$. By
means of this theorem we are able to compactly represent
$1$-dimensional size functions as formal series of cornerpoints
and cornerlines (An example is given by Figure
\ref{DistMatch22}$(a)$ and Figure \ref{DistMatch22}$(b)$).
\end{remark}

As a first and simple consequence of the Representation Theorem
\ref{RepresentationTheorem}, we have the following result, that
will be useful in Section \ref{MainResults} (cf. \cite{FrLa01}):

\begin{cor}\label{CorDisc}
Each discontinuity point $(\bar x,\bar y)$ for
$\ell_{(\M,\varphi)}$ is such that either $\bar x$ is a
discontinuity point for $\ell_{(\M,\varphi)}(\cdot,\bar y)$, or
$\bar y$ is a discontinuity point for $\ell_{(\M,\varphi)}(\bar
x,\cdot)$, or both these conditions hold.
\end{cor}

We are now able to introduce the matching distance $d_{match}$.
Before going on, we observe that the Representation Theorem
$\ref{RepresentationTheorem}$ allows us to reduce the problem of
comparing $1$-dimensional size functions into the comparison of
the related multisets of cornerpoints. Indeed, the matching
distance $d_{match}$ can be seen as a measure of the cost of
transporting the cornerpoints of a $1$-dimensional size function
into the cornerpoints of another one, with respect to a functional
$\delta$ depending on the $L_{\infty}$-distance between two
matched cornerpoints and on their $L_{\infty}$-distance from the
diagonal $\{(x,y)\in\R^2:x=y\}$. An example of matching between
two formal series is given by Figure \ref{DistMatch22}$(c)$.

\begin{figure}
\psfrag{(A)}{{\large $(a)$}}\psfrag{(B)}{{\large
$(b)$}}\psfrag{(C)}{{\large $(c)$}}
\psfrag{x}{{\large$x$}}\psfrag{y}{\!{\large$y$}}\psfrag{optimal}{\!\!\!\!matching}\psfrag{matching}{}
\begin{center}
\includegraphics[width=\textwidth]{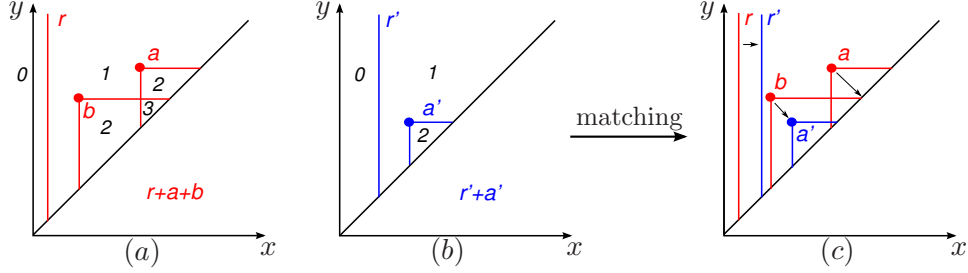}
\end{center}
\caption{$(a)$ Size function corresponding to the formal series
$r+a+b$. $(b)$ Size function corresponding to the formal series
$r'+a'$. $(c)$ The matching between the two formal series,
realizing the matching distance between the two size
functions.}\label{DistMatch22}
\end{figure}

Let us now define more formally the matching distance $d_{match}$.
Assume that two $1$-dimensional size functions $\ell_1$, $\ell_2$
are given. Consider the multiset $C_1$  (respectively $C_2$) of
cornerpoints for $\ell_1$ (resp. $\ell_2$), counted with their
multiplicities and augmented by adding the points of the diagonal
$\{(x,y)\in\R^2:x=y\}$ counted with infinite multiplicity. If we
denote by $\bar\Delta^*$ the set $\bar\Delta^+$ extended by the
points at infinity of the kind $(a,\infty)$, i.e.
$\bar\Delta^*=\bar\Delta^+\cup\{(a,\infty):a\in\R\}$, the matching
distance $d_{match}\left(\ell_1,\ell_2\right)$ is then defined as

$$
d_{match}\left(\ell_1,\ell_2\right)=\min_{\sigma}\max_{P\in C_1}\delta(P,\sigma(P)),
$$

where $\sigma$ varies among all the bijections between $C_1$ and $C_2$ and

$$
\delta((x,y),(x',y'))=\min\left\{\max\left\{|x-x'|,|y-y'|\right\},\max\left\{\frac{y-x}{2},\frac{y'-x'}{2}\right\}\right\},
$$
for every $(x,y)$, $(x',y')$ $\in\bar\Delta^*$ and with the
convention about $\infty$ that $\infty-y=y-\infty=\infty$ when
$y\neq\infty$, $\infty-\infty=0$, $\frac{\infty}{2}=\infty$,
$|\infty|=\infty$, $\min\{c,\infty\}=c$ and
$\max\{c,\infty\}=\infty$.

In plain words, the pseudometric $\delta$ measures the
pseudodistance between two points $(x,y)$ and $(x',y')$ as the
minimum between the cost of moving one point onto the other and
the cost of moving both points onto the diagonal, with respect to
the max-norm and under the assumption that any two points of the
diagonal have vanishing pseudodistance (we recall that a
pseudodistance $d$ is just a distance missing the condition
$d(X,Y)=0\Rightarrow X=Y$, i.e. two distinct elements may have
vanishing distance with respect to $d$).

An application of the matching distance is given by Figure
\ref{DistMatch22}$(c)$. As can be seen by this example, different
$1$-dimensional size functions may in general have a different
number of cornerpoints. Therefore $d_{match}$ allows a proper
cornerpoint to be matched to a point of the diagonal: this
matching can be interpreted as the destruction of a proper
cornerpoint.  Moreover, we stress that the matching distance is
stable with respect to perturbations of the measuring functions,
as the following Matching Stability Theorem states:

\begin{theorem}[Matching Stability Theorem]\label{StabilityTheorem}
If $(\M,\varphi)$, $(\M,\psi)$ are two size pairs with $\max_{P\in\M}|\varphi(P)-\psi(P)|\leq\varepsilon$, then it holds that
$d_{match}(\ell_{(\M,\varphi)},\ell_{(\M,\psi)})\leq\varepsilon$.
\end{theorem}

For a proof of the previous theorem and more details about the
matching distance the reader is referred to
\cite{dAFrLabis,dAFrLa} (see also \cite{CoEdHa07} for the analogue
of the matching distance in Persistent Homology and its
stability).

\subsubsection{Coordinates of cornerpoints and discontinuity points}

Following the related literature (see also \cite{DiLa} for the
case of measuring functions with a finite number of critical
homological values), it can be easily deduced that, if finite,
both the coordinates of a cornerpoint for a $1$-dimensional size
function $\ell_{(\M,\varphi)}$ are critical values of the
measuring function $\varphi$, under the assumption that $\varphi$
is $C^1$. However, to the best of our knowledge, this result has
never been explicitly proved until now. Therefore, for the sake of
completeness we formalize here this statement, that will be used
in Section \ref{MainResults}:
\begin{theorem}\label{ThNuovo}
Let $\M$ be a closed $C^1$ Riemannian manifold, and let
$\varphi:\M\to\R$ be a $C^1$ measuring function. Then if $(\bar
x,\bar y)$ is a proper cornerpoint for $\ell_{(\M,\varphi)}$, it
follows that both $\bar x$ and $\bar y$ are critical values of
$\varphi$. If $(\bar x,\infty)$ is a cornerpoint at infinity for
$\ell_{(\M,\varphi)}$, it follows that $\bar x$ is a critical
value of $\varphi$.
\end{theorem}
\begin{proof}
We confine ourselves to prove the former statement, since the
proof of the latter is analogous.

First of all, let us remark that there exists a closed
$C^{\infty}$ Riemannian manifold $\widetilde\M$ that is
$C^1$-diffeomorphic to $\M$ through a $C^1$-diffeomorphism
$h:\widetilde\M\to\M$ (cf. \cite[Thm. 2.9]{Hi76}). Set
$\tilde\varphi=\varphi\circ h$. Obviously, the size functions
associated with the size pairs $(\widetilde\M,\tilde\varphi)$ and
$(\M,\varphi)$ coincide. Therefore, $(\bar x,\bar y)$ is also a
cornerpoint for $\ell_{(\widetilde\M,\tilde\varphi)}$.

We observe that the claim of our theorem holds for a closed
$C^{\infty}$ Riemannian manifold endowed with a Morse measuring
function (see \cite[Thm. 2.2]{Fr96}). Now, for every real value
$\varepsilon>0$ it is possible to find a Morse measuring function
$\varphi_{\varepsilon}:\widetilde{\M}\to\R$ such that
$\max_{Q\in\widetilde{\M}}|\tilde\varphi(Q)-\varphi_{\varepsilon}(Q)|\leq\varepsilon$
and
$\max_{Q\in\widetilde{\M}}\left\|\nabla\tilde\varphi(Q)-\nabla\varphi_{\varepsilon}(Q)\right\|\leq\varepsilon$:
We can obtain $\varphi_{\varepsilon}$ by considering first the
smooth measuring function given by the convolution of
$\tilde\varphi$ and an opportune ``regularizing'' function, and
then a Morse measuring function $\varphi_{\varepsilon}$
approximating in $C^1(\widetilde\M,\R)$ the previous measuring
function (cf. \cite[Corollary 6.8]{Mi63}). Therefore, from the
Matching Stability Theorem \ref{StabilityTheorem} it follows that
for every $\varepsilon>0$ we can find a cornerpoint $(\bar
x_{\varepsilon},\bar y_{\varepsilon})$ for the size function
$\ell_{(\widetilde{\M},\varphi_{\varepsilon})}$ with $\left\|(\bar
x,\bar y)-(\bar x_{\varepsilon},\bar
y_{\varepsilon})\right\|_{\infty}\leq\varepsilon$ and $\bar
x_{\varepsilon},\bar y_{\varepsilon}$ as critical values for
$\varphi_{\varepsilon}$. Passing to the limit for $\varepsilon\to
0$ we obtain that both $\bar x$ and $\bar y$ are critical values
for $\tilde\varphi$. The claim follows by observing that, since
$\tilde\varphi$ and $\varphi$ have the same critical values, both
$\bar x$ and $\bar y$ are also critical values for $\varphi$.
\end{proof}

From the Representation Theorem \ref{RepresentationTheorem} and
Theorem \ref{ThNuovo} we can obtain the following corollary,
refining Corollary \ref{CorDisc} in the $C^1$ case (we skip the
easy proof):

\begin{cor}\label{DiscR}
Let $\M$ be a closed $C^1$ Riemannian manifold, and let $\varphi:\M\to\R$ be
a $C^1$ measuring function. Let also $(\bar x,\bar y)$ be a
discontinuity point for $\ell_{(\M,\varphi)}$.
Then at least one of the following statements holds:
\begin{description}
  \item[(i)] $\bar x$
is a discontinuity point for $\ell_{(\M,\varphi)}(\cdot,\bar y)$ and
$\bar x$ is a critical value for $\varphi$;
  \item[(ii)] $\bar y$
is a discontinuity point for $\ell_{(\M,\varphi)}(\bar
x,\cdot)$ and $\bar y$ is a critical value for $\varphi$.
\end{description}
\end{cor}


The generalization of Corollary \ref{DiscR} in the $k$-dimensional
setting is not so simple and requires some new ideas which are
given in Section \ref{MainResults}, which also provides our main
results.

\subsection{Reduction to the 1-dimensional case}\label{Reduction}

We are now ready to review the approach to multidimensional Size
Theory proposed in \cite{BiCeFrGiLa08}. In that work, the authors
prove that the case $k>1$ can be reduced to the $1$-dimensional
framework by a change of variable and the use of a suitable
foliation. In particular, they show that there exists a
parameterized family of half-planes in $\R^k\times\R^k$ such that
the restriction of a $k$-dimensional size function
$\ell_{(\mathcal{M},\fr)}$ to each of these half-planes can be
seen as a particular $1$-dimensional size function. The
motivations at the basis of this approach move from the fact that
the concepts of proper cornerpoint and cornerpoint at infinity,
defined for $1$-dimensional size functions, appear not easily
generalizable to an arbitrary dimension (namely the case $k>1$).
As a consequence, at a first glance it does not seem possible to obtain
the multidimensional analogue of the matching distance $d_{match}$
and therefore it is not clear how to generalize the Matching Stability Theorem
\ref{StabilityTheorem}. On the other hand, all
these problems can be bypassed by means of the results we recall
in the rest of this subsection.

\begin{definition}\label{Adm}
\label{np} For every unit vector $\vec{l}=(l_1,\ldots,l_k)$ of
$\mathbb{R}^k$ such that $l_i>0$ for $i=1,\dots,k$, and for every
vector $\vec{b}=(b_1,\ldots,b_k)$ of $\mathbb{R}^k$ such that
$\sum_{i=1}^k b_i=0$, we shall say that the pair $(\vec{l},\vec{b})$
is \emph{admissible}. We shall denote the set of all admissible
pairs in $\R^k\times\R^k$ by $Adm_k$. Given an admissible pair
$(\vec{l},\vec{b})$, we define the half-plane
$\pi_{(\vec{l},\vec{b})}$ of $\R^k\times\R^k$ by the following
parametric equations:
$$
\left\{%
\begin{array}{ll}
    \vec x=s\vec l + \vec b\\
    \vec y=t\vec l + \vec b\\
\end{array}%
\right.
$$
for $s,t\in \R$, with $s<t$.
\end{definition}

The following proposition implies that the collection of
half-planes given in Definition \ref{Adm} is actually a foliation
of $\Delta^+$.

\begin{prop}\label{Foliazione}
For every $(\vec{x},\vec{y})\in \Delta^+$ there exists one and
only one admissible pair $(\vec{l},\vec{b})$ such that
$(\vec{x},\vec{y})\in \pi_{(\vec{l},\vec{b})}$.
\end{prop}

Now we can show the reduction to the $1$-dimensional case.

\begin{theorem}[Reduction Theorem]\label{reduction}
Let $(\vec{l},\vec{b})$ be an admissible pair, and $F_{(\vec l,\vec
b)}^{\fr}:\mathcal{M}\rightarrow\R$ be defined by setting
$$
F_{(\vec l,\vec
b)}^{\fr}(P)=\max_{i=1,\dots,k}\left\{\frac{\varphi_i(P)-b_i}{l_i}\right\}\
.
$$
Then, for every $(\vec x,\vec y)=(s\vec l+\vec b,t\vec l + \vec
b)\in\pi_{(\vec{l},\vec{b})}$ the following equality holds:
$$
\ell_{(\mathcal{M},\fr)}(\vec x,\vec y)=\ell_{(\mathcal{M},F_{(\vec
l,\vec b)}^{\fr})}(s,t)\ .
$$
\end{theorem}

In the following, we shall use the symbol $F_{(\vec l,\vec
b)}^{\fr}$ in the sense of the Reduction Theorem \ref{reduction}.

\begin{remark}\label{Descriptor}
In plain words, the Reduction Theorem \ref{reduction} states that
each multidimensional size function corresponds to a
$1$-dimensional size function on each half-plane of the given
foliation. It follows that each multidimensional size function can
be represented as a parameterized family of formal series of
points and lines, following the description introduced in
Subsection \ref{PartCase1} for the case $k=1$. Indeed, it is
possible to associate a formal series $\sigma_{(\vec l,\vec b)}$
with each admissible pair $(\vec l,\vec b)$, with $\sigma_{(\vec
l,\vec b)}$ describing the $1$-dimensional size function
$\ell_{(\mathcal{M},F_{(\vec l,\vec b)}^{\fr})}$. Therefore, on
each half-plane $\pi_{(\vec l,\vec b)}$, the matching distance
$d_{match}$ and the Matching Stability Theorem
\ref{StabilityTheorem} can be applied. Moreover, the family
$\left\{\sigma_{(\vec l,\vec b)}:(\vec l,\vec b)\in Adm_k\right\}$
turns out to be a complete descriptor for $\ell_{(\mathcal{M},\vec
\varphi)}$, since two multidimensional size functions coincide if
and only if the corresponding parameterized families of formal
series coincide.
\end{remark}

Before proceeding, we now introduce an example showing how the
Reduction Theorem \ref{reduction} works.

\begin{example}\label{Example}

In $\R^3$ consider the set
$\mathcal{Q}=[-1,1]\times[-1,1]\times[-1,1]$ and the unit sphere
$S^2$ of equation $x^2+y^2+z^2=1$. Let also
$\vec\Phi=(\Phi_1,\Phi_2):\R^3\rightarrow\R^2$ be the continuous
function, defined as $\vec\Phi(x,y,z)=(|x|,|z|)$. In this setting,
consider the size pairs $(\mathcal{M},\fr)$ and
$(\mathcal{N},\vec{\psi})$ where
$\mathcal{M}=\partial\mathcal{Q}$, $\mathcal{N}=S^2$, and $\fr$
and $\vec\psi$ are respectively the restrictions of $\vec\Phi$ to
$\mathcal{M}$ and $\mathcal{N}$. In order to compare the size
functions $\ell_{(\mathcal{M},\fr)}$ and
$\ell_{(\mathcal{N},\vec\psi)}$, we are interested in studying the
foliation in half-planes $\pi_{(\vec l,\vec b)}$, where $\vec
l=(\cos\theta,\sin\theta)$ with $\theta\in(0,\frac{\pi}{2})$, and
$\vec b=(a,-a)$ with $a\in\R$. Any such half-plane is represented
by
$$
\left\{%
\begin{array}{ll}
    x_1=s\cos\theta + a\\
    x_2=s\sin\theta - a\\
    y_1=t\cos\theta + a\\
    y_2=t\sin\theta - a\\
\end{array}%
\right.\ ,
$$
with $s,t\in\R$, $s<t$. Figure \ref{figura:FigArt} shows the size
functions $\ell_{(\mathcal{M},F_{(\vec l,\vec b)}^{\fr})}$ and
$\ell_{(\mathcal{N},F_{(\vec l,\vec b)}^{\vec\psi})}$, for
$\theta=\frac{\pi}{4}$ and $a=0$, i.e. $\vec
l=\left(\frac{\sqrt{2}}{2},\frac{\sqrt{2}}{2}\right)$ and $\vec
b=(0,0)$.
\begin{figure}[h]
\begin{center}
\psfrag{R2}{\!\!\!{\small$\!\!t=\mathbf{\sqrt{2}}$}}
\psfrag{R1}{\!\!{\small$t=\mathbf{1}$}} \psfrag{M}{$\mathcal{M}$}
\psfrag{N}{$\mathcal{N}$} \psfrag{C}{$\ell_{(\mathcal{M},F_{(\vec
l,\vec b)}^{\fr})}$} \psfrag{D}{\ $\ell_{(\mathcal{N},F_{(\vec
l,\vec b)}^{\vec\psi})}$}\psfrag{x}{$x$}\psfrag{y}{$y$}\psfrag{z}{$z$}
\psfrag{s}{$s$}\psfrag{t}{$t$}
\begin{tabular}{c}
\includegraphics[width=\textwidth]{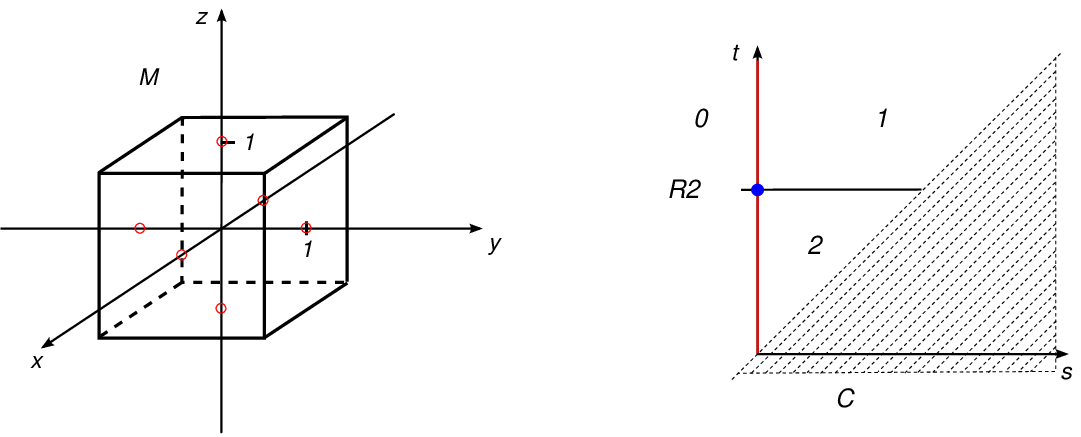}\\
\includegraphics[width=\textwidth]{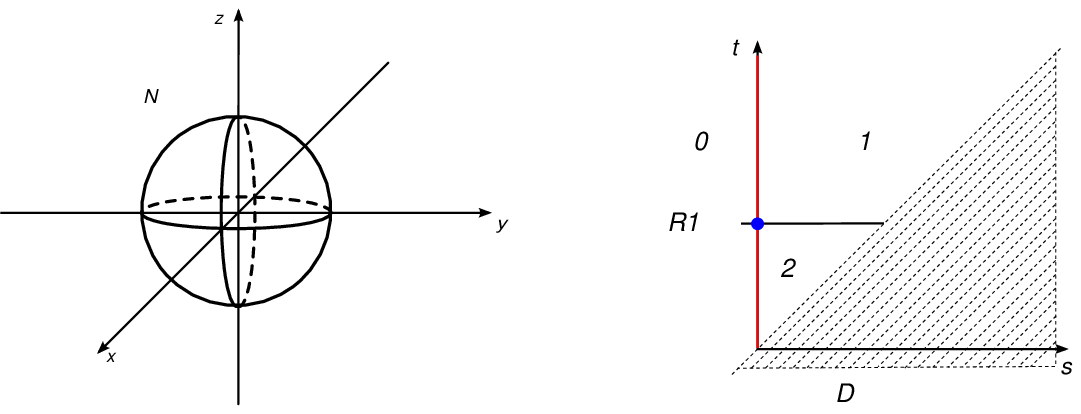}
\end{tabular}
\caption{The topological spaces $\mathcal{M}$ and $\mathcal{N}$
and the size functions $\ell_{(\mathcal{M},F_{(\vec l,\vec
b)}^{\fr})},\ell_{(\mathcal{N},F_{(\vec l,\vec b)}^{\vec\psi})}$
associated with the half-plane $\pi_{(\vec l,\vec b)}$, for $\vec
l=(\frac{\sqrt{2}}{2},\frac{\sqrt{2}}{2})$ and $\vec b=(0,0)$.}
\label{figura:FigArt}
\end{center}
\end{figure}
With this choice, we obtain that $F_{(\vec l,\vec
b)}^{\fr}=\sqrt{2}\max\{\varphi_1,\varphi_2\}=\sqrt{2}\max\{|x|,|z|\}$
and $F_{(\vec l,\vec
b)}^{\vec\psi}=\sqrt{2}\max\{\psi_1,\psi_2\}=\sqrt{2}\max\{|x|,|z|\}$.
Therefore, Theorem \ref{reduction} implies that, for every
$(x_1,x_2,y_1,y_2)\in\pi_{(\vec l,\vec b)}$, we have
\begin{eqnarray*}
\ell_{(\mathcal{M},\fr)}(x_1,x_2,y_1,y_2)=\ell_{(\mathcal{M},\fr)}\left(\frac{s}{\sqrt{2}},\frac{s}{\sqrt{2}},\frac{t}{\sqrt{2}},\frac{t}{\sqrt{2}}\right)=\ell_{(\mathcal{M},F_{(\vec l,\vec b)}^{\fr})}(s,t)\\
\ell_{(\mathcal{N},\vec\psi)}(x_1,x_2,y_1,y_2)=\ell_{(\mathcal{N},\vec\psi)}\left(\frac{s}{\sqrt{2}},\frac{s}{\sqrt{2}},\frac{t}{\sqrt{2}},\frac{t}{\sqrt{2}}\right)=\ell_{(\mathcal{N},F_{(\vec
l,\vec b)}^{\vec\psi})}(s,t)\ .
\end{eqnarray*}
The matching distance $d_{match}(\ell_{(\mathcal{M},F_{(\vec
l,\vec b)}^{\fr})},\ell_{(\mathcal{N},F_{(\vec l,\vec
b)}^{\vec\psi})})$ is equal to $\sqrt{2}-1$, i.e. the cost of
moving the point of coordinates $(0,\sqrt{2})$ onto the point of
coordinates $(0,1)$, computed with respect to the $\max$-norm. The
points $(0,\sqrt{2})$ and $(0,1)$ are representative of the
characteristic triangles of the size functions
$\ell_{(\mathcal{M},F_{(\vec l,\vec b)}^{\fr})}$ and
$\ell_{(\mathcal{N},F_{(\vec l,\vec b)}^{\vec\psi})}$,
respectively. Note that the matching distance
computed for $\vec
l=\left(\frac{\sqrt{2}}{2},\frac{\sqrt{2}}{2}\right)$ and $\vec
b=(0,0)$ induces a pseudodistance. This means that, even by
considering just one half-plane of the foliation, it is possible
to effectively compare multidimensional size functions. We
conclude by observing that
$\ell_{(\mathcal{M},\varphi_1)}\equiv\ell_{(\mathcal{N},\psi_1)}$
and
$\ell_{(\mathcal{M},\varphi_2)}\equiv\ell_{(\mathcal{N},\psi_2)}$.
In other words, the multidimensional size functions, with respect
to $\fr,\vec\psi$, are able to discriminate the cube and the
sphere, while both the $1$-dimensional size functions, with
respect to $\varphi_1,\varphi_2$ and $\psi_1,\psi_2$, cannot do
that. This higher discriminatory power of multidimensional size
functions gives a further motivation for their definition and use.
\end{example}

The next result proves the stability of $d_{match}$ with respect to
the choice of the half-planes of the foliation. Indeed, the next
proposition states that small enough changes in $(\vec l, \vec b)$
with respect to the $\max$-norm induce small changes of
$\ell_{(\mathcal{M},F_{(\vec l,\vec b)}^{\vec\varphi})}$ with
respect to the matching distance.

\begin{prop}\label{stability2}
If $(\mathcal{M},\vec\varphi)$ is a size pair, $(\vec l,\vec b)\in
Adm_k$ and $\varepsilon$ is a real number with
$0<\varepsilon<\min_{i=1,\dots,k} l_i$, then for every admissible
pair $(\vec l',\vec b')$ with $\|(\vec l,\vec b)-(\vec l',\vec
b)\|_{\infty}\leq\varepsilon$, it holds that
$$d_{match}(\ell_{(\mathcal{M},F_{(\vec l,\vec
b)}^{\vec\varphi})},\ell_{(\mathcal{M},F_{(\vec l',\vec
b')}^{\vec\varphi})})\leq\varepsilon\cdot\frac{\max_{P\in\mathcal{M}}\left\|\vec\varphi(P)\right\|_{\infty}\!+\!\|\vec
l\|_{\infty}\!+\!\|\vec
b\|_{\infty}}{\min_{i=1,\dots,k}\{l_i(l_i-\varepsilon)\}}.$$
\end{prop}

\begin{remark}\label{stability}
Analogously, it is possible to prove (cf.
\cite[Prop.~2]{BiCeFrGiLa08}) that $d_{match}$ is stable with
respect to the chosen measuring function, i.e. that small enough
changes in $\fr$ with respect to the $\max$-norm induce small
changes of $\ell_{(\mathcal{M},F_{(\vec l,\vec b)}^{\fr})}$ with
respect to the matching distance.
\end{remark}

Proposition \ref{stability2} and Remark \ref{stability} guarantee
the stability of this approach.

\section{Main Results}\label{MainResults}

In this section we are going to prove some new results
about the discontinuities of multidimensional size functions.
In order to do that, we will confine ourselves to
the case of a size pair $(\mathcal{M},\vec\varphi)$, where
$\mathcal{M}$ is a closed $C^1$ Riemannian $m$-manifold.

From now to Theorem \ref{Teoremone} we shall assume that an
admissible pair $(\vec l,\vec b)\in Adm_k$ is fixed, considering
the $1$-dimensional size function $\ell_{(\M,F)}$, where
$F(Q)=\max_{i=1,\dots,k}\frac{\varphi_i(Q)-b_i}{l_i}$. We shall
say that $F$ and $\ell_{(\M,F)}$ are the ($1$-dimensional)
measuring function and the size function corresponding to the
half-plane $\pi_{(\vec l,\vec b)}$, respectively.

The main results of this section are stated in Theorem
\ref{Teoremone} and Theorem \ref{Teoremone2}, showing a necessary
condition for a point $(\vec x,\vec y)\in\Delta^+$ to be a
discontinuity point for the size function
$\ell_{(\M,\vec\varphi)}$, under the assumption that $\vec\varphi$
is $C^1$ and $C^0$, respectively. For the sake of clarity, we will
now provide  a sketch of the arguments that will lead us to the
proof of our main results.

Theorem \ref{Teoremone} is a generalization in the $k$-dimensional
setting of Corollary \ref{DiscR}, stating that each discontinuity
point for a $1$-dimensional size function $\ell_{(\M,\varphi)}$,
related to a $C^1$ measuring function $\varphi$, is such that at
least one of its coordinates is a critical value for $\varphi$. We
recall that Corollary \ref{DiscR} directly descends from the
Representation Theorem \ref{RepresentationTheorem} and from
Theorem \ref{ThNuovo}, according to which each finite coordinate
of a cornerpoint for $\ell_{(\M,\varphi)}$ has to be a critical
value for $\varphi$. Our first goal is to prove that a modified
version of this last statement holds for the $1$-dimensional size
function $\ell_{(\M,F)}$ corresponding to the half-plane
$\pi_{(\vec l,\vec b)}$. The reason for such an adaptation is that
the $1$-dimensional measuring function $F$ is not $C^1$ (even in
case $\fr$ is $C^1$), and therefore we need to generalize the
concepts of critical point and critical value by introducing the
definitions of {\em $(\vec l,\vec b)$-pseudocritical point} and
{\em $(\vec l,\vec b)$-pseudocritical value} for a $C^1$ function
(Definition \ref{PcPoint}). These notions, together with an
approximation in $C^0(\M,\R)$ of the function $F$ by $C^1$
functions, are used to prove that, if $\fr\in C^1(\M,\R^k)$, each
finite coordinate of a cornerpoint for $\ell_{(\M,F)}$ has to be
an $(\vec l,\vec b)$-pseudocritical value for $\vec\varphi$
(Theorem \ref{main}). Next, we show (Proposition \ref{LemmaDisc})
that a correspondence exists between the discontinuity points of
$\ell_{(M,F)}$ and the ones of $\ell_{(M,\vec\varphi)}$. Theorem
\ref{main} and Proposition \ref{LemmaDisc} lead us to the relation
(Theorem \ref{TeorDisc}) between the discontinuity points for
$\ell_{(\M,\vec\varphi)}$, lying on the half-plane $\pi_{(\vec
l,\vec b)}$, and the $(\vec l,\vec b)$-pseudocritical values for
$\vec\varphi$. This last result is refined in Theorem
\ref{Teoremone} under the assumption that $\fr$ is $C^1$,
providing a necessary condition for discontinuities of
$\ell_{(\M,\vec\varphi)}$ that does not depend on the half-planes
of the foliation. This can be done by introducing the concepts of
{\em pseudocritical point} and {\em pseudocritical value} for an
$\R^k$-valued $C^1$ function (Definition \ref{PseudoCritico}), and
considering a suitable projection $\rho:\R^k\to\R^h$. The
necessary condition given in Theorem \ref{Teoremone} is finally
generalized to the case of continuous measuring functions (Theorem
\ref{Teoremone2}), once more by means of an approximation
technique, and the notions of \emph{special point} and
\emph{special value}.

Before going on, we need the following definition:

\begin{definition}\label{PcPoint}
Assume that $\fr\in C^1(\M,\R^k)$. For every $Q\in\mathcal{M}$,
set
$I_Q=\left\{i\in\{1,\dots,k\}:\frac{\varphi_i(Q)-b_i}{l_i}=F(Q)\right\}$.
We shall say that $Q$ is an {\em $(\vec l,\vec b)$-pseudocritical
point for $\vec\varphi$} if the convex hull of the gradients
$\nabla\varphi_i(Q)$, $i\in I_Q$, contains the null vector, i.e.
for every $i\in I_Q$ there exists a real value $\lambda_i$ such
that  $\sum_{i\in I_Q}\lambda_i\nabla\varphi_i(Q)=\mathbf 0$, with
$0\leq\lambda_i\leq1$ and $\sum_{i\in I_Q}\lambda_i=1$.  If $Q$ is
an $(\vec l,\vec b)$-pseudocritical point for $\vec\varphi$, the
value $F(Q)$ will be called an {\em $(\vec l,\vec
b)$-pseudocritical value for $\vec\varphi$}.
\end{definition}

\begin{remark}
The concept of $(\vec l,\vec b)$-pseudocritical point is strongly
connected, via the function $F$ introduced in Definition
\ref{PcPoint}, with the notion of generalized gradient introduced
by F.~H.~Clarke \cite{Cl90}. For a point $Q\in\M$, the condition
of being $(\vec l,\vec b)$-pseudocritical for $\vec\varphi$
corresponds to the one of being ``critical'' for the generalized
gradient of $F$ \cite[Prop.~2.3.12]{Cl90}. However, in this
context we prefer to adopt a terminology highlighting the dependence
on the considered half-plane.
\end{remark}

We can now state our first result.

\begin{theorem}\label{main}
Assume that $\fr\in C^1(\M,\R^k)$. If $(\sigma,\tau)$ is a proper
cornerpoint of $\ell_{(\M,F)}$, then both $\sigma$ and $\tau$ are
${(\vec l,\vec b\,)}$-pseudocritical values for $\vec\varphi$. If
$(\sigma,\infty)$ is a cornerpoint at infinity of $\ell_{(\M,F)}$,
then $\sigma$ is an ${(\vec l,\vec b\,)}$-pseudocritical value for
$\vec\varphi$.
\end{theorem}
\begin{proof}
We confine ourselves to proving the former statement, since the
proof of the latter is analogous. The idea is to show that our
thesis holds for a $C^1$ function approximating the measuring
function $F:\mathcal{M}\to\R$ in $C^0(\M,\R)$, and verify that
this property passes to the limit. Let us now set
$\Phi_i(Q)=\frac{\varphi_i(Q)-b_i}{l_i}$ and choose $c\in\R$ such
that $\min_{Q\in\M}\Phi_i(Q)>-c$, for every $i=1,\dots,k$.
Consider the function sequence $(F_p)$,
$p\in\mathbb{N}^+=\mathbb{N}\setminus\{0\}$, where
$F_p:\mathcal{M}\to\R$ and
$F_p(Q)=\left(\sum_{i=1}^k(\Phi_i(Q)+c)^p\right)^{\frac{1}{p}}-c$:
Such a sequence converges uniformly  to the function $F$. Indeed,
for every $Q\in\mathcal{M}$ and for every index $p$ we have that
{\setlength\arraycolsep{2pt}
\begin{eqnarray*}
|F(Q)-F_p(Q)|&=&\Bigg|\max_i\Phi_i(Q)-\Bigg(\Bigg(\sum_{i=1}^k(\Phi_i(Q)+c)^p\Bigg)^{\frac{1}{p}}-c\Bigg)\Bigg|=\\
&=&\Bigg|\max_i\{\Phi_i(Q)+c\}-\Bigg(\sum_{i=1}^k(\Phi_i(Q)+c)^p\Bigg)^{\frac{1}{p}}\Bigg|=\\
&=&\Bigg(\sum_{i=1}^k(\Phi_i(Q)+c)^p\Bigg)^{\frac{1}{p}}-\max_i\{\Phi_i(Q)+c\}\leq\\&\leq&\max_i\{\Phi_i(Q)+c\}\cdot(k^{\frac{1}{p}}-1).
\end{eqnarray*}}
Let us now consider a proper cornerpoint $\bar C$ of the size
function $\ell_{(\mathcal{M},F)}$. By the Matching Stability
Theorem \ref{StabilityTheorem} it follows that it is possible to
find a large enough $p$ and a proper cornerpoint $C_p$ of the
$1$-dimensional size function $\ell_{(\mathcal{M},F_p)}$
(associated with the size pair $(\M,F_p)$) such that $C_p$ is
arbitrarily close to $\bar C$. Since $C_p$ is a proper cornerpoint
of $\ell_{(\mathcal{M},F_p)}$, it follows from Theorem
\ref{ThNuovo} that its coordinates are critical values of the
$C^1$ function $F_p$. By focusing our attention on the abscissa of
$C_p$ (analogous considerations hold for the ordinate of $C_p$) it
follows that there exists $Q_p\in\mathcal{M}$ with
$x(C_p)=F_p(Q_p)$ and (in respect to local coordinates
$x_1,\dots,x_m$ of the $m$-manifold $\M$)
{\setlength\arraycolsep{2pt}
\begin{eqnarray*}
0&=&\frac{\partial F_p}{\partial
x_1}(Q_p)=\left(\sum_{i=1}^{k}(\Phi_i(Q_p)+c)^{p}\right)^{\frac{1-p}{p}}\cdot\left(\sum_{i=1}^{k}(\Phi_i(Q_p)+c)^{p-1}\cdot\frac{\partial
\Phi_i}{\partial x_1}(Q_p)\right)\\
\vdots&&\\
0&=&\frac{\partial F_p}{\partial
x_m}(Q_p)=\left(\sum_{i=1}^{k}(\Phi_i(Q_p)+c)^{p}\right)^{\frac{1-p}{p}}\cdot\left(\sum_{i=1}^{k}(\Phi_i(Q_p)+c)^{p-1}\cdot\frac{\partial
\Phi_i}{\partial x_m}(Q_p)\right)\ .
\end{eqnarray*}}

Hence we have {\setlength\arraycolsep{2pt}
\begin{eqnarray*}
\sum_{i=1}^{k}&&(\Phi_i(Q_p)+c)^{p-1}\cdot\frac{\partial
\Phi_i}{\partial x_1}(Q_p)=0\\
\vdots\ \ &&\\
\sum_{i=1}^{k}&&(\Phi_i(Q_p)+c)^{p-1}\cdot\frac{\partial
\Phi_i}{\partial x_m}(Q_p)=0\ .
\end{eqnarray*}}

Therefore, by setting
$$\boldsymbol{v}_p=(v_p^1,\dots,v_p^k)=\left((\Phi_1(Q_p)+c)^{p-1},\dots,(\Phi_k(Q_p)+c)^{p-1}\right),$$
we can write $^t\!\!J(Q_p)\cdot
^t\!\boldsymbol{v}_p=\textbf{0}$, where $J(Q_p)$ is the Jacobian
matrix of $\vec\Phi=(\Phi_1,\dots,\Phi_k)$ computed at the point
$Q_p$. By the compactness of $\mathcal{M}$, we can assume
(possibly by extracting a subsequence) that $(Q_p)$ converges to a
point $\bar Q$. Let us define
$\boldsymbol{u}_p=\frac{\boldsymbol{v}_p}{\left\|\boldsymbol{v}_p\right\|_{\infty}}$. By
compactness (recall that $\left\|\boldsymbol{u}_p\right\|_{\infty}=1$)  we can also assume
(possibly by considering a subsequence) that the sequence $(\boldsymbol{u}_p)$
converges to a vector $\boldsymbol{\bar u}=(\bar u^1,\dots,\bar u^k)$, where
$\bar{u}^i=\lim_{p\to\infty}\frac{v_p^i}{\left\|\boldsymbol{v}_p\right\|_{\infty}}$ and $\left\|\boldsymbol{\bar u}\right\|_{\infty}=1$.
Obviously $^t\!\!J(Q_p)\cdot ^t\!\!\boldsymbol{u}_p=\textbf{0}$ and hence we
have
\begin{eqnarray}\label{jacobiano}
^t\!\!J(\bar Q)\cdot ^t\!\!\boldsymbol{\bar u}=\textbf{0}.
\end{eqnarray}
Since for every index $p$ and for every $i=1,\dots,k$ the relation
$0<u_p^i\leq 1$ holds, for each $i=1,\dots,k$ the condition
$0\leq\bar{u}^i=\lim_{p\to\infty}u_p^i\leq1$ is satisfied. Let us
now recall that $F(\bar Q)=\max_i\Phi_i(\bar Q)$, by definition,
and consider the set $I_{\bar Q}=\{i\in\{1,\dots,k\}:\Phi_i({\bar
Q})=F({\bar Q})\}=\{i_1,\dots,i_h\}$. For every $r\not\in I_{\bar
Q}$ the component  $\bar u^r$ is equal to $0$, since $0\leq
u_p^r=\left(\frac{\Phi_r\left(Q_p\right)+c}{\max_i\{\Phi_r\left(Q_p\right)+c\}}\right)^{p-1}$
and
$\lim_{p\to\infty}\frac{\Phi_r\left(Q_p\right)+c}{\max_i\{\Phi_r\left(Q_p\right)+c\}}=\frac{\Phi_r\left(\bar
Q\right)+c}{F\left(\bar Q\right)+c}$, which is strictly less than
$1$ for $\Phi_r(\bar Q)<F(\bar Q)$. Hence we have
$\boldsymbol{\bar u}=\bar u^{i_1}\cdot\boldsymbol{e}_{i_1}+\dots
+\bar u^{i_h}\cdot\boldsymbol{e}_{i_h}$, where $\boldsymbol{e}_i$
is the $i^{th}$ vector of the standard basis of $\R^k$. Thus, from
equality (\ref{jacobiano}) we have $\sum_{j=1}^{h}\bar
u^{i_j}\cdot\frac{\partial\Phi_{i_j}}{\partial x_1}(\bar
Q)=0,\dots,\sum_{j=1}^{h}\bar
u^{i_j}\cdot\frac{\partial\Phi_{i_j}}{\partial x_m}(\bar Q)=0$,
that is $\sum_{j=1}^{h}\frac{\bar
u^{i_j}}{l_{i_j}}\cdot\frac{\partial\varphi_{i_j}}{\partial
x_1}(\bar Q)=0,\dots,\sum_{j=1}^{h}\frac{\bar
u^{i_j}}{l_{i_j}}\cdot\frac{\partial\varphi_{i_j}}{\partial
x_m}(\bar Q)=0$, since $\Phi_i=\frac{\varphi-b_i}{l_i}$. Hence,
$\sum_{j=1}^h\frac{\bar u^{i_j}}{l_{i_j}}\nabla\varphi_{i_j}(\bar
Q)=\textbf{0}$. By recalling that $\bar u^{i_j}\geq 0$,
$l_{i_j}>0$ and $\boldsymbol{\bar u}$ is a non--vanishing vector,
it follows immediately that $\sum_{j=1}^h\frac{\bar
u^{i_j}}{l_{i_j}}>0$ and therefore the convex hull of the
gradients $\nabla\varphi_{i_1}({\bar Q}),\dots,
\nabla\varphi_{i_h}({\bar Q})$ contains the null vector. Thus,
$\bar Q$ is an $(\vec l,\vec b)$-pseudocritical point for
$\vec\varphi$ and hence $F(\bar Q)$ is an $(\vec l,\vec
b)$-pseudocritical value for $\vec\varphi$. Moreover, from the
uniform convergence of the sequence $(F_p)$ to $F$ and from the
continuity of the function $F$, we have (recall that $\bar
C=\lim_{p\to\infty}C_p$)
$$x(\bar{C})=\lim_{p\to \infty}x(C_p)=\lim_{p\to \infty}F_p(Q_p)=F(\bar Q).$$
In other words, the abscissa $x(\bar C)$ of a proper cornerpoint
of $\ell_{(\mathcal{M},F)}$ is the image of an $(\vec l,\vec
b)$-pseudocritical point $\bar Q$ through $F$, i.e. an $(\vec
l,\vec b)$-pseudocritical value for $\vec\varphi$. An analogous
reasoning holds for the ordinate $y(\bar C)$ of a proper
cornerpoint.
\end{proof}

Our next result shows that each discontinuity of
$\ell_{(\mathcal{M},\vec\varphi)}$ corresponds to a discontinuity
of the $1$-dimensional size function associated with a suitable
half-plane of the foliation.

\begin{prop}\label{LemmaDisc}
A point $(\vec x,\vec y)=(s\cdot\vec l+\vec b, t\cdot\vec l+\vec
b)\in\pi_{(\vec l,\vec b\,)}$ is a discontinuity point for
$\ell_{(\M,\vec\varphi)}$ if and only if $(s,t)$ is a
discontinuity point for $\ell_{(\M,F)}$.
\end{prop}
\begin{proof}
Obviously, if $(s,t)$ is a discontinuity point for
$\ell_{(\M,F)}$, then $(\vec x,\vec y)=(s\cdot\vec l+\vec b,
t\cdot\vec l+\vec b)\in\pi_{(\vec l,\vec b\,)}$ is a discontinuity
point for $\ell_{(\M,\vec\varphi)}$, because of the Reduction
Theorem \ref{reduction}. In order to prove the inverse
implication, we shall verify the contrapositive statement, i.e. if
$(s,t)$ is not a discontinuity point for $\ell_{(\M,F)}$, then
$(s\cdot\vec l+\vec b,t\cdot\vec l+\vec b)$ is not a discontinuity
point for $\ell_{(\M,\vec\varphi)}$. Indeed, if $(s,t)$ is not a
discontinuity point for $\ell_{(\M,F)}$, then $\ell_{(\M,F)}$ is
locally constant at $(s,t)$ (recall that each size function is
natural--valued). Therefore it will be possible to choose a real
number
 $\eta>0$ such that
\begin{eqnarray}\label{equal}
\ell_{(\M,F)}(s-\eta,t+\eta)=\ell_{(\M,F)}(s+\eta,t-\eta).
\end{eqnarray}

Before proceeding in our proof, we need the following result:

\begin{lemma}\label{BoundDMatch}
Let $(\mathcal{M},\psi)$, $(\mathcal{M},\psi')$ be two size pairs,
with $\psi,\psi':\mathcal{M}\to\mathbb{R}$. If
$d_{match}\left(\ell_{(\mathcal{M},\psi)},\ell_{(\mathcal{M},\psi')}\right)\leq
2\varepsilon$, then it holds that
$$
\ell_{(\mathcal{M},\psi)}(s-\varepsilon,t+\varepsilon)\leq\ell_{(\mathcal{M},\psi')}(s+\varepsilon,t-\varepsilon),
$$
for every $(s,t)$ with $s+\varepsilon< t-\varepsilon$.
\end{lemma}
\begin{proof}[Proof of Lemma \ref{BoundDMatch}]
Let $\Delta^*$ be the set given by
$\Delta^+\cup\{(a,\infty):a\in\R\}$. For every $(s,t)$ with $s<t$,
let us define the set
$L_{(s,t)}=\{(\sigma,\tau)\in\Delta^*:\sigma\leq s,\tau>t\}$. By
the Representation Theorem \ref{RepresentationTheorem} we have
that $\ell_{(\mathcal{M},\psi)}(s-\varepsilon,t+\varepsilon)$
equals the number of proper cornerpoints and cornerpoints at
infinity for $\ell_{(\mathcal{M},\psi)}$ belonging to the set
$L_{(s-\varepsilon,t+\varepsilon)}$. Since
$d_{match}\left(\ell_{(\mathcal{M},\psi)},\ell_{(\mathcal{M},\psi')}\right)\leq
2\varepsilon$, the number of proper cornerpoints and cornerpoints
at infinity for $\ell_{(\mathcal{M},\psi')}$ in the set
$L_{(s+\varepsilon,t-\varepsilon)}$ is not less than
$\ell_{(\mathcal{M},\psi)}(s-\varepsilon,t+\varepsilon)$. The
reason is that the change from $\psi$ to $\psi'$ does not move the
cornerpoints more than $2\varepsilon$, with respect to the
$\max$-norm, because of the Matching Stability Theorem
\ref{StabilityTheorem}. By applying the Representation Theorem
\ref{RepresentationTheorem} once again to
$\ell_{(\mathcal{M},\psi')}$, we get our thesis.
\end{proof}

Let us go back to the proof of Proposition \ref{LemmaDisc}. By
Proposition \ref{stability2}, we can then consider a real value
$\varepsilon=\varepsilon(\eta)$ with
$0<\varepsilon<\min_{i=1,\dots,k}l_i$ such that for every
admissible pair $(\vec l',\vec b')$ with $\left\|(\vec l,\vec b)-(\vec
l',\vec b')\right\|_{\infty}\leq\varepsilon$, the relation
$d_{match}(\ell_{(\mathcal{M},F)},\ell_{(\mathcal{M},F')})\leq\frac{\eta}{2}$
holds, where $\ell_{(\M,F')}$ is the $1$-dimensional size function
corresponding to the half-plane $\pi_{(\vec l',\vec b')}$. By
applying Lemma \ref{BoundDMatch} twice and the monotonicity of
$\ell_{(\M,F')}$ in each variable (cf. Remark \ref{monotone}), we get the
inequalities{\setlength\arraycolsep{2pt}
\begin{eqnarray}\label{inequal}
\ell_{(\M,F)}(s-\eta,
t+\eta)&\leq&\ell_{(\M,F')}(s-\frac{\eta}{2},t+\frac{\eta}{2})\nonumber\\
&\leq&\ell_{(\M,F')}(s+\frac{\eta}{2},t-\frac{\eta}{2})\leq\ell_{(\M,F)}(s+\eta,t-\eta).
\end{eqnarray}}
Because of equality (\ref{equal}) we have that the inequalities
(\ref{inequal}) imply {\setlength\arraycolsep{2pt}
\begin{eqnarray}\label{equal2}
\ell_{(\M,F)}(s-\eta,
t+\eta)&=&\ell_{(\M,F')}(s-\frac{\eta}{2},t+\frac{\eta}{2})\nonumber\\
&=&\ell_{(\M,F')}(s+\frac{\eta}{2},t-\frac{\eta}{2})=\ell_{(\M,F)}(s+\eta,t-\eta).
\end{eqnarray}}
Therefore, once again because of the monotonicity of
$\ell_{(\M,F')}$ in each variable, for every $(s',t')$ with
$\left\|(s,t)-(s',t')\right\|_{\infty}\leq\frac{\eta}{2}$ and for
every $(\vec l',\vec b')$ with $\|(\vec l,\vec b)-(\vec l',\vec
b')\|_{\infty}\leq\varepsilon$ the equality
$\ell_{(\M,F')}(s',t')=\ell_{(\M,F)}(s,t)$ holds. By applying the
Reduction Theorem \ref{reduction} we get
$\ell_{(\M,\vec\varphi)}(s'\cdot\vec l'+\vec b',t'\cdot\vec
l'+\vec b')=\ell_{(\M,\vec\varphi)}(s\cdot\vec l+\vec b,t\cdot\vec
l+\vec b)$. In other words, $\ell_{(\M,\vec\varphi)}$ is locally
constant at the point $(\vec x,\vec y)$, and hence $(\vec x,\vec
y)$ is not a discontinuity point for $\ell_{(\M,\vec\varphi)}$.
\end{proof}

\begin{remark}\label{RemLemmaDisc}
Let us observe that Proposition \ref{LemmaDisc} holds under weaker
hypotheses, i.e. in the case that $\M$ is a non-empty, compact and
locally connected Hausdorff space. However, for the sake of
simplicity, we prefer here to confine ourselves to the setting
assumed at the beginning of the present section.
\end{remark}

The following theorem associates the discontinuities of a
multidimensional size function to the $(\vec l,\vec
b)$-pseudocritical values of $\vec\varphi$.

\begin{theorem}\label{TeorDisc}
Let $(\vec x,\vec y)\in\Delta^+$ with $(\vec x,\vec y)=(s\cdot\vec
l+\vec b, t\cdot\vec l+\vec b)\in\pi_{(\vec l,\vec b\,)}$. If $(\vec x,\vec y)$ is a discontinuity point for
$\ell_{(\M,\vec\varphi)}$ then at
least one of the following statements holds:
\begin{description}
  \item[$(i)$] $s$
is a discontinuity point for $\ell_{(\M,F)}(\cdot,t)$;
  \item[$(ii)$] $t$ is a discontinuity point for
  $\ell_{(\M,F)}(s,\cdot)$.
\end{description}
Moreover, $(i)$ and $(ii)$ are equivalent to
\begin{description}
  \item[$(i')$] $\vec x$
is a discontinuity point for $\ell_{(\M,\vec\varphi)}(\cdot,\vec
y)$;
  \item[$(ii')$] $\vec y$ is a discontinuity point for $\ell_{(\M,\vec\varphi)}(\vec
x,\cdot)$,
\end{description}
respectively. If $\fr\in C^1(\M,\R^k)$, statement $(i)$ implies
that $s$ is an $(\vec l,\vec b)$-pseudocriti\-cal value for
$\vec\varphi$, and statement $(ii)$ implies that $t$ is an $(\vec
l,\vec b)$-pseudocritical value for $\vec\varphi$.
\end{theorem}
\begin{proof}
By Proposition \ref{LemmaDisc} we have that $(s,t)$ is a
discontinuity point for $\ell_{(\M,F)}$, and from Corollary
\ref{CorDisc} it follows that either $s$ is a discontinuity point
for $\ell_{(\M,F)}(\cdot,t)$ or $t$ is a discontinuity point for
$\ell_{(\M,F)}(s,\cdot)$, or both these conditions hold, thus
proving the first part of the theorem.

Let us now suppose that $s$ is a discontinuity point for
$\ell_{(\M,F)}(\cdot,t)$. Since the function
$\ell_{(\M,F)}(\cdot,t)$ is monotonic, then for every real value
$\varepsilon>0$ we have that
$\ell_{(\M,F)}(s-\varepsilon,t)\neq\ell_{(\M,F)}(s+\varepsilon,t)$.
Moreover, the following equalities hold because of the Reduction
Theorem \ref{reduction}: {\setlength\arraycolsep{2pt}
\begin{eqnarray}\label{Uguaglianze}\begin{array}{ccc}
\ell_{(M,F)}(s-\varepsilon,t)&=&\ell_{(M,\vec\varphi)}((s-\varepsilon)\cdot\vec
l+\vec b,t\cdot\vec l+\vec b)=\ell_{(M,\vec\varphi)}(\vec
x-\varepsilon\cdot\vec l,\vec y)\\
\ell_{(M,F)}(s+\varepsilon,t)&=&\ell_{(M,\vec\varphi)}((s+\varepsilon)\cdot\vec
l+\vec b,t\cdot\vec l+\vec b)=\ell_{(M,\vec\varphi)}(\vec
x+\varepsilon\cdot\vec l,\vec y).\end{array}
\end{eqnarray}} By setting $\vec\varepsilon=\varepsilon\cdot\vec l$, we
get $\ell_{(M,\vec\varphi)}(\vec x-\vec\varepsilon,\vec
y)\neq\ell_{(M,\vec\varphi)}(\vec x+\vec\varepsilon,\vec y)$.
Therefore $\vec x$ is a discontinuity point for
$\ell_{(M,\vec\varphi)}(\cdot,\vec y)$, thus proving that
$(i)\Rightarrow(i')$.

Let us now prove that $(i')\Rightarrow(i)$. If $\vec x$ is a
discontinuity point for $\ell_{(M,\vec\varphi)}(\cdot,\vec y)$,
from the monotonicity in the variable $\vec x$ (cf. Remark
\ref{monotone}) it follows that $\ell_{(M,\vec\varphi)}(\vec
x-\varepsilon\cdot\vec l,\vec y)\neq\ell_{(M,\vec\varphi)}(\vec
x+\varepsilon\cdot\vec l,\vec y)$ for every $\varepsilon>0$.
Therefore, because of the equalities (\ref{Uguaglianze}) we get
$\ell_{(\M,F)}(s-\varepsilon,t)\neq\ell_{(\M,F)}(s+\varepsilon,t)$,
proving that $(i')\Rightarrow(i)$. Analogously, we can show that
$(ii)\Leftrightarrow(ii')$.

Furthermore, if $s$ is a discontinuity
point for $\ell_{(\M,F)}(\cdot,t)$, from the Representation
Theorem \ref{RepresentationTheorem} it follows that $s$ is the
abscissa of a cornerpoint (possibly at infinity). Hence, if $\fr\in C^1(\M,\R^k)$ then by
Theorem \ref{main} we have that $s$ is an $(\vec l,\vec
b)$-pseudocritical value for $\vec\varphi$.

In a similar way, we can examine the case that $t$ is a
discontinuity point for $\ell_{(\M,F)}(s,\cdot)$, and get the
final statement.
\end{proof}

Before giving the first of our main results, we need the following
definition.

\begin{definition}\label{PseudoCritico}
Let $\vec\xi:\mathcal{M}\to\mathbb{R}^h$, and suppose that
$\vec\xi$ is $C^1$ at a point $Q\in\M$.
The point $Q$ is said to be a {\em pseudocritical point for
$\vec\xi$} if the convex hull of the gradients $\nabla\xi_i(Q)$,
$i=1,\dots,h$, contains the null vector, i.e. there exist
$\lambda_1,\dots,\lambda_h\in\mathbb{R}$ such that
$\sum_{i=i}^{h}\lambda_i\cdot\nabla\xi_i(Q)=\mathbf 0$, with
$0\leq\lambda_i\leq1$ and $\sum_{i=1}^{h}\lambda_i=1$. If $Q$ is a
pseudocritical point of $\vec\xi$, then $\vec\xi(Q)$ will be
called a {\em pseudocritical value for $\vec\xi$}.
\end{definition}

\begin{remark}
Definition \ref{PseudoCritico} corresponds to the Fritz John
necessary condition for optimality in Nonlinear Programming
\cite{BaShSh93}. We shall use the term ``pseudocritical'' just for
the sake of conciseness. For further references see \cite{Sm75}.
The concept of pseudocritical point is strongly related also
to the one of Jacobi Set (cf. \cite{HeHa02}).
\end{remark}

The next example makes Definition \ref{PseudoCritico} clearer.

\begin{example}\label{EsempioPcp}
Let us compute the pseudocritical points and values for the
measuring function $\vec\xi=(\xi_1,\xi_2):\M\to\R^2$, where $\M$
is the surface coinciding with the unit sphere $S^2\subset\R^3$,
and $\vec\xi$ is obtained as the restriction to $\M$ of the
function $\vec\Xi=(\Xi_1,\Xi_2):\R^3\to\R^2$, with
$\vec\Xi(x,y,z)=(x,z)$ (see Figure \ref{Pcp}). According to
Definition \ref{PseudoCritico}, it follows that a point $Q\in\M$
is pseudocritical for $\vec\xi$ if and only if either
$\nabla\xi_1(Q)=\mathbf 0$, or $\nabla\xi_2(Q)=\mathbf 0$, or these two gradient
vectors are parallel with opposite verse. Referring to our
example, $\nabla\xi_1(Q)$ and $\nabla\xi_2(Q)$ are the projections
of $\nabla\Xi_1(Q)=(1,0,0)$ and $\nabla\Xi_2(Q)=(0,0,1)$ onto the
tangent space of $\M$ at $Q$, respectively. Therefore, it can be
easily verified that the pseudocritical points of $\M$ for the
function $\vec\xi$ are given by the set
$\{(\cos\alpha,0,\sin\alpha),\ 0\leq\alpha\leq\frac{\pi}{2}\ \vee\
\pi\leq\alpha\leq\frac{3}{2}\pi\}$. Hence, the corresponding
pseudocritical values are the elements of the set
$\{(\cos\alpha,\sin\alpha),\ 0\leq\alpha\leq\frac{\pi}{2}\ \vee\
\pi\leq\alpha\leq\frac{3}{2}\pi\}$.
\end{example}

\begin{figure}[h]
\psfrag{x}{$x$}\psfrag{y}{$y$}\psfrag{z}{$z$}\psfrag{P}{$Q$}
\psfrag{chi1(Q)}{\footnotesize{$\nabla\Xi_1(Q)$}}
\psfrag{chi2(Q)}{\footnotesize{$\nabla\Xi_2(Q)$}}
\psfrag{Chi1(Q)}{\footnotesize{$\nabla\xi_1(Q)$}}
\psfrag{Chi2(Q)}{\footnotesize{$\nabla\xi_2(Q)$}}
\includegraphics[width=\textwidth]{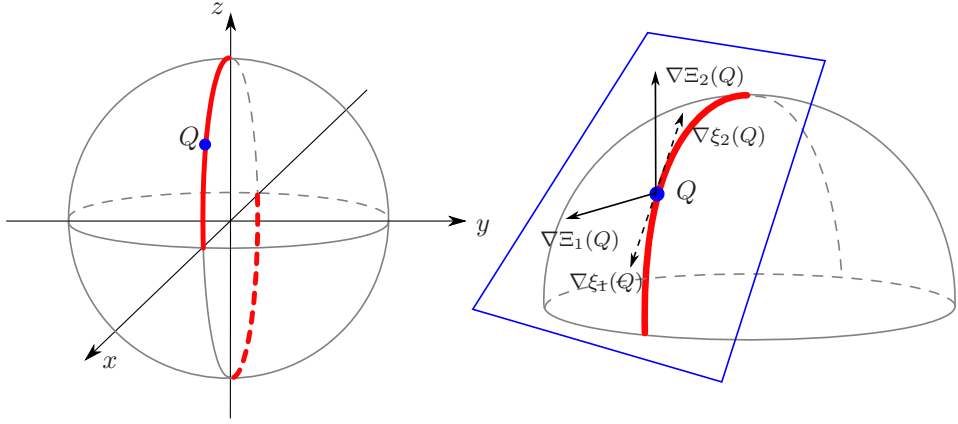}
\caption{$(a)$ The sphere $S^2\subseteq\R^3$ endowed with the
measuring function $\vec\xi=(\xi_1,\xi_2):S^2\to\R^2$, defined as
$\vec\xi(x,y,z)=(x,z)$ for each $(x,y,z)\in S^2$. The
pseudocritical points of $\vec\xi$ are depicted in bold red. $(b)$
The point $Q$ is a pseudocritical point for $\vec\xi$, since the
vectors $\nabla\xi_1(Q)$ and $\nabla\xi_2(Q)$ are parallel with
opposite verse.}\label{Pcp}
\end{figure}

In the following, we shall say that
$\rho:\mathbb{R}^k\to\mathbb{R}^h$ is a {\em projection} if there
exist $h$ indices $i_1,\dots,i_h$ such that
$\rho((x_1,\dots,x_k))=(x_{i_1},\dots,x_{i_h})$, for every $\vec
x=(x_1,\dots,x_k)\in\mathbb{R}^k$. In other words, such a function
$\rho$ is used to delete some components of a vector $\vec
x\in\R^k$.


We are now ready to give the first main result of this paper.

\begin{theorem}\label{Teoremone}
Assume that $\fr\in C^1(\M,\R^k)$. Let $(\vec x,\vec
y)\in\Delta^+$ be a discontinuity point for
$\ell_{(\M,\vec\varphi)}$. Then at least one of the following
statements holds:
\begin{description}
    \item[$(i)$] $\vec x$ is a discontinuity point for
$\ell_{(\M,\vec\varphi)}(\cdot,\vec y)$;
    \item[$(ii)$] $\vec y$ is a discontinuity point for
$\ell_{(\M,\vec\varphi)}(\vec x,\cdot)$.
\end{description}
Moreover, if $(i)$ holds, then a projection $\rho$ exists such
that $\rho(\vec x)$ is a pseudocritical value for
$\rho\circ\vec\varphi$. If $(ii)$ holds, then a projection $\rho$
exists such that $\rho(\vec y)$ is a pseudocritical value for
$\rho\circ\vec\varphi$.
\end{theorem}
\begin{proof}
Because of Proposition \ref{Foliazione}, an admissible pair $(\vec
l,\vec b)$ exists, such that $(\vec x,\vec y)=(s\cdot\vec l+\vec
b,t\cdot\vec l+\vec b)$ for a suitable pair $(s,t)$. Statements
$(i)$ and $(ii)$ are guaranteed by Theorem \ref{TeorDisc},
assuring that either $\vec x$ is a discontinuity point for
$\ell_{(\M,\vec\varphi)}(\cdot,\vec y)$ and $s$ is an $(\vec
l,\vec b)$-pseudocritical value for $\vec\varphi$, or $\vec y$ is
a discontinuity point for $\ell_{(\M,\vec\varphi)}(\vec x,\cdot)$
and $t$ is an $(\vec l,\vec b)$-pseudocritical value for
$\vec\varphi$, or both these conditions hold.

Let us now confine ourselves to assume that $\vec x$ is a
discontinuity point for $\ell_{(\M,\vec\varphi)}(\cdot,\vec y)$
and $s$ is an $(\vec l,\vec b)$-pseudocritical value for
$\vec\varphi$. We shall prove that a projection $\rho$ exists such
that $\rho(\vec x)$ is a pseudocritical value for
$\rho\circ\vec\varphi$. The proof in the case that $\vec y$ is a
discontinuity point for $\ell_{(\M,\vec\varphi)}(\vec x,\cdot)$
and $t$ is an $(\vec l,\vec b)$-pseudocritical value for
$\vec\varphi$ proceeds in quite a similar way. Since $s$ is an
$(\vec l,\vec b)$-pseudocritical value for $\vec\varphi$, by
Definition \ref{PcPoint} there exist a point $Q\in\mathcal{M}$ and
some indices $i_{1},\dots,i_{h}$ with $1\leq h\leq k$, such that
$s=F(Q)=\frac{\varphi_{i_1}(Q)-b_{i_1}}{l_{i_1}}=
\dots=\frac{\varphi_{i_h}(Q)-b_{i_h}}{l_{i_h}}$ and
$\sum_{j=1}^h\lambda_j\cdot\nabla\vec\varphi_{i_j}(Q)=\mathbf 0$,
with $0\leq\lambda_j\leq 1$ for $j=1,\dots,h$, and
$\sum_{j=1}^{h}\lambda_j=1$. Let us now consider the projection
$\rho:\mathbb{R}^k\to\mathbb{R}^h$ defined by setting $\rho(\vec
x)=(x_{i_1},\dots,x_{i_h})$. Since $(\vec x,\vec
y)=(x_1,\dots,x_k,y_1,\dots,y_k)=(s\cdot l_1+b_1,\dots,s\cdot l_k+
b_k, t\cdot l_1+ b_1,\dots,t\cdot l_k+ b_k)$, we observe that
$x_{i_j}=\left(\frac{\varphi_{i_j}(Q)-b_{i_j}}{l_{i_j}}\right)\cdot
l_{i_j}+b_{i_j}=\varphi_{i_j}(Q)$, for every $j=1,\dots,h$.
Therefore it follows that $\rho(\vec x)$ is a pseudocritical value
for $\rho\circ\vec\varphi$.
\end{proof}

\begin{remark}
We stress that Theorem \ref{Teoremone} improves the result
obtained in Theorem \ref{TeorDisc}, providing a necessary
condition for discontinuities of multidimensional size functions
that does not depend on the foliation of the domain $\Delta^+$.
\end{remark}


\subsection{Refining Theorem \ref{Teoremone} to less regular measuring
functions}

In this section we generalize Theorem \ref{Teoremone} to the case
of continuous measuring functions. 
%
In what follows, we shall call a \emph{special point for a
continuous function $\vec\xi:\M\to\R^h$} any point $Q\in\M$ where
$\vec\xi$ is not $C^1$. If $Q$ is a special point for $\vec\xi$,
the value $\vec\xi(Q)$ will be called a \emph{special value for
$\vec\xi$}.
%

\begin{theorem}\label{Teoremone2}
Let $(\vec x,\vec y)\in\Delta^+$ be a discontinuity point for
$\ell_{(\M,\vec\varphi)}$. Then at least one of the following
statements holds:
\begin{description}
    \item[$(i)$] $\vec x$ is a discontinuity point for
$\ell_{(\M,\vec\varphi)}(\cdot,\vec y)$;
    \item[$(ii)$] $\vec y$ is a discontinuity point for
$\ell_{(\M,\vec\varphi)}(\vec x,\cdot)$.
\end{description}
Moreover, if $(i)$ holds, then a projection $\rho$ exists such
that $\rho(\vec x)$ is either a special value or a pseudocritical
value for $\rho\circ\vec\varphi$. If $(ii)$ holds, then a
projection $\rho$ exists such that $\rho(\vec y)$ is either a
special value or a pseudocritical value for
$\rho\circ\vec\varphi$.

\end{theorem}
\begin{proof}
Because of Proposition \ref{Foliazione}, an admissible pair $(\vec
l,\vec b)$ exists, such that $(\vec x,\vec y)=(s\cdot\vec l+\vec
b,t\cdot\vec l+\vec b)$ for a suitable pair $(s,t)$. Statements
$(i)$ and $(ii)$ are guaranteed by Theorem \ref{TeorDisc},
assuring that either $\vec x$ is a discontinuity point for
$\ell_{(\M,\vec\varphi)}(\cdot,\vec y)$ and $s$
is a discontinuity point for $\ell_{(\M,F)}(\cdot,t)$, or $\vec y$ is
a discontinuity point for $\ell_{(\M,\vec\varphi)}(\vec x,\cdot)$
and $t$
is a discontinuity point for $\ell_{(\M,F)}(s,\cdot)$, or both these conditions hold.

Let us now assume that $\vec x$ is a discontinuity point for
$\ell_{(\M,\vec\varphi)}(\cdot,\vec y)$ and $s$
is a discontinuity point for $\ell_{(\M,F)}(\cdot,t)$.
We shall prove that a projection $\rho$ exists such that
$\rho(\vec x)$ is either a special value or a pseudocritical value
for $\rho\circ\vec\varphi$.

Call $\mathbb{S}_j$ the set of special points of $\varphi_j:\M\to\R$, for $j=1,\ldots,k$. For
every $i\in\mathbb{N}^+=\mathbb{N}\setminus\{0\}$ and $j=1,\ldots,k$, consider the
compact set $K^i_j=\{Q\in\M:d(Q,\mathbb{S}_j)\geq\frac{1}{i}\}$, and
take a $C^1$ function
$\varphi^i_j:\M\to\R$ such that
\begin{enumerate}
\item[(1)]
$\max_{Q\in\mathcal{M}}|\varphi_j(Q)-\varphi^i_j(Q)|\leq\frac{1}{i}$;
\item[(2)] $\max_{Q\in
K^i_j}\|\nabla\varphi_j(Q)-\nabla\varphi_j^i(Q)\|\leq\frac{1}{i}$.
\end{enumerate}
This can be done by considering the convolution of each component
$\varphi_j$, $j=1,\dots,k$, with a suitable ``regularizing''
function.

From now on, for the sake of conciseness we shall use the symbols
$F$ and $F^i$ to denote the functions  $F_{(\vec l\,\vec
b)}^{\fr}=\max_{j=1,\dots,k}\left\{\frac{\varphi_j-b_j}{l_j}\right\}$
and $F_{(\vec l\,\vec
b)}^{\fr^i}=\max_{j=1,\dots,k}\left\{\frac{\varphi_j^i-b_j}{l_j}\right\}$,
respectively. For
every $i\in\mathbb{N}^+$, we also set $\fr^i=(\varphi^i_1,\ldots,\varphi^i_k)$.

Since $s$ is a discontinuity point for $\ell_{(\M,F)}(\cdot,t)$,
by the Representation Theorem \ref{RepresentationTheorem} it
follows that a cornerpoint of $\ell_{(\M,F)}$ (proper or at
infinity) of coordinates $(s,\bar t)$ exists, with $\bar t>t$.
Moreover, by condition $(1)$ we have that the sequence $(F^i)$
uniformly converges to $F$. Therefore, the Matching Stability
Theorem \ref{StabilityTheorem} implies that a sequence $((s^i,\bar
t^i))$ exists, such that $(s^i,\bar t^i)$ is a cornerpoint for
$\ell_{(\M,F^i)}$ and $((s^i,\bar t^i))$ converges to $(s,\bar
t)$. For every large enough index $i$, once more by the
Representation Theorem \ref{RepresentationTheorem}, $s^i$ is then
a discontinuity point for $\ell_{(\M,F^i)}(\cdot,t)$, and hence by
Theorem \ref{TeorDisc} we have that $\vec x^{\,i}=s^i\cdot\vec
l+\vec b$ is a discontinuity point for
$\ell_{(\M,\fr^i)}(\cdot,\vec y)$. From Theorem \ref{Teoremone} it
follows that a projection $\rho^i$ exists, such that $\rho^i(\vec
x^{\,i})$ is a pseudocritical value for
$\rho^i\circ\vec\varphi^i$. Possibly by considering a subsequence,
we can suppose that all the $\rho^i$ equal a projection $\rho$.
Moreover, we can consider a sequence $(Q^i)$ such that
$Q^i\in\M$, $\rho\circ\fr^{i}(Q^i)=\rho(\vec x^{\,i})$ and $Q^i$ is a
pseudo-critical point for $\rho\circ\fr^i$. Furthermore, by the
compactness of $\M$, possibly by extracting a subsequence we can
assume $(Q^i)$ converging to a point $Q\in\M$. From the continuity
of $\fr$ and from the uniform convergence of $(\fr^i)$ to $\fr$,
we can deduce
\begin{enumerate}
\item[(3)]
$\rho\circ\fr(Q)=\lim_{i\to\infty}\rho\circ\fr(Q^i)=\lim_{i\to\infty}\rho\circ\fr^i(Q^i)=\lim_{i\to\infty}\rho(\vec
x^{\,i})=\rho(\vec x)$.
\end{enumerate}
If $\rho(\vec x)$ is a special value for $\rho\circ\fr$ then our claim is proved. 
If $\rho(\vec x)=(x_{j_1},\ldots,x_{j_h})$ is not a special value for $\rho\circ\fr$ then $Q\not\in \mathbb{S}_{j_1}\cup\ldots\cup\mathbb{S}_{j_h}$. Hence, for any
large enough index $i$, it follows that $Q,Q^i\in K^i_{j_1}\cap\ldots\cap K^i_{j_h}$.  By recalling
that each point $Q^i$ is a pseudocritical point for
$\rho\circ\fr^i$, and by observing that the property of being a
pseudocritical point passes to the limit, we get that $\rho(\vec
x)$ is a pseudocritical value for $\rho\circ\vec\varphi$. In other
words, we have just proved that if $\vec x$ is a discontinuity
point for $\ell_{(\M,\fr)}(\cdot,\vec y)$, then a projection
$\rho$ exists such that $\rho(\vec x)$ is either a special value
or a pseudocritical value for $\rho\circ\vec\varphi$.

Analogously, it is possible to prove that if $\vec y$ is a
discontinuity point for $\ell_{(\M,\fr)}(\vec x,\cdot)$, then a
projection $\rho$ exists such that $\rho(\vec y)$ is either a
special value or a pseudocritical value for
$\rho\circ\vec\varphi$.
\end{proof}

\subsection{Consequences of our results}

The results proved in this paper imply several relevant
consequences. First of all they contribute to clarifying the
structure of multidimensional size functions. In order to explain
this point let us consider the case of a compact smooth surface
$\mathcal{S}$ endowed with a smooth function $\fr : \mathcal{S}\to
\R^2$. It is immediate to verify that all pseudocritical points
belong to the Jacobi set of $\fr$, that is the set where the
gradients $\nabla\varphi_1$ and $\nabla\varphi_2$ are parallel.
This implies (cf. \cite{HeHa02}) that in the generic case the
pseudocritical points belong to a $1$-submanifold $\mathcal{J}$ of
$\mathcal{S}$ (in local coordinates such a manifold is determined
by the vanishing of the Jacobian of $\fr$). Now, Theorem
\ref{Teoremone2} allows us to claim that all discontinuity points
$(x_1,x_2,y_1,y_2)$ of the size function $\ell_{(\M,\vec\varphi)}$
belong either to $\mathcal{J}\times \R^2$ or to
$\R^2\times\mathcal{J}$. For the computation of $\mathcal{J}$ we
refer to \cite{HeHa02}.

In the light of this new information, we can imagine the possibility of constructing new
algorithms to efficiently compute multidimensional size functions.
Let us consider the connected components in which the domain of
$\ell_{(\M,\vec\varphi)}$ is divided by the set
$(\mathcal{J}\times \R^2)\cup (\R^2\times\mathcal{J})$. Since size
functions are locally constant at each point of continuity (we
recall that they are natural-valued), we immediately obtain that
$\ell_{(\M,\vec\varphi)}$ is constant at each of those connected
components. It follows that the computation of
$\ell_{(\M,\vec\varphi)}$ just requires the computation of its
value at only one point for each connected component. These
observations open the way to new and more efficient methods of
computation for multidimensional size functions.

Our results also make new pseudodistances between size functions
computable in an easier way. Indeed, let us consider two size
pairs $(\M,\vec\varphi)$, $(\mathcal{N},\vec\psi)$ and the value
$\delta_H$ giving the Hausdorff distance between the sets where
$\ell_{(\M,\vec\varphi)}$ and $\ell_{(\mathcal{N},\vec\psi)}$ are
discontinuous. It is trivial to check that the function $d_D$
defined by setting
$d_D\left(\ell_{(\M,\vec\varphi)},\ell_{(\mathcal{N},\vec\psi)}\right)=\delta_H$
is a pseudodistance between multidimensional size functions.
Helping us to localize the discontinuities of multidimensional
size functions, Theorem \ref{Teoremone2} makes the computation of
$d_D$ easier.

\section*{Conclusions and future work}

In this paper we have proved that a discontinuity point for a
multidimensional size function has at least one special or
pseudocritical coordinate, under the hypothesis that the
considered measuring function is (at least) continuous. This
result is a first step in the development of Size Theory for
$\R^k$-valued measuring functions. Indeed, the localization of the
unique points where $k$-dimensional size functions can be
discontinuous allows us to better understand Topological
Persistence and opens the way to the formulation of effective
algorithms for its computation. On the other hand, it is worth
noting that our framework could be applicable also to the study of
discontinuities in persistent algebraic topology, including
persistent homology groups and size homotopy groups. However, some
difficulties could derive from the present lack of the analogue of
Theorem \ref{StabilityTheorem} for those structures, i.e. a
stability result in the case of {\em continuous} (possibly
non-tame \cite{CoEdHa07}) measuring functions. These last research
lines appear to be promising, both from the theoretical and the
applicative point of view.

\subsection*{Acknowledgements}
Work performed within the activity of ARCES ``E. De Castro'',
University of Bologna, under the auspices of INdAM-GNSAGA.

The authors thank Davide
Guidetti (University of Bologna) for his helpful advice.

This paper is dedicated to Martina and Riccardo.

\renewcommand{\thesection}{A}
\setcounter{equation}{0}  

\section{Appendix}\label{Appendix}

\subsection{Relationship between Size Theory and Persistent Homology}

Size Theory and Persistent Homology are deeply connected theories.
We shall recall some similarities and differences between them in
this appendix. For more details we refer to the survey papers
\cite{BiDeFaFrGiLaPaSp08} and \cite{EdHa08}.

Size Theory was born at the beginning on the 1990s (cf.
\cite{Fr90,Fr91}) as a mathematical approach to shape comparison.
The main idea is to describe shape as a pseudometric (the
\emph{natural pseudodistance}) between topological spaces endowed
with real-valued functions, called \emph{measuring functions}. The
measuring functions are used as descriptors of the properties with
respect to which the topological spaces are compared. For example,
if we are interested in the comparison of two objects $A$ and $B$
with respect to their bumps and hollows, it can be natural to
consider two subsets of $\R^3$ representing their bodies, endowed
with two functions associating with each point its distance from the
center of mass of the body it belongs to. On the other hand, if we
are interested in the colorings of $A$ and $B$, we can consider
two surfaces endowed with functions representing the color taken
at each point. The natural pseudodistance between two pairs
\emph{(topological space, measuring function)}, called \emph{size
pairs}, is the infimum of the change of the measuring function
under the action of all possible homeomorphisms from one
topological space to the other. Size functions and size homotopy
groups (their algebraic-topological equivalent; cf. \cite{FrMu99})
appeared as mathematical tools useful for computing lower bounds for
the natural pseudodistance, introducing \emph{ante litteram} the
study of Topological Persistence.

Persistent Homology was born approximately ten years later, at the
beginning of the 00s, as a mathematical approach to studying the
homology of topological spaces known just by a sampling. In this
case, the attention was focused on the radius $r$ of the spheres
centered at the sample points, whose union approximates the
topological space. The problem of choosing the value of $r$ led to the concept
of persistence, emphasizing the topological properties stable
under the change of $r$. In other words, the main goal was
topological simplification, in order to get the relevant
topological information concerning the object under study. The
value $r$, playing the role of the measuring function in Size
Theory, has been subsequently extended to more general functions.

Despite their different origins and goals, Size Theory and
Persistent Homology have developed similar structures and
concepts, under different names. In order to help readers who are 
not familiar with both these theories, this
section compares some of their key concepts, explaining
their reciprocal links. These connections and
relationships are summarized in Table~\ref{t1}.

\begin{table}[h]

 \begin{tabular}{|c|c|c|c|}
  \hline
  \textbf{Size Theory}          & \textbf{references}         & \textbf{Persistent Homology}   & \textbf{references}      \\
  \hline\hline
  \emph{size pair}              & \cite{Fr90,Fr96}            & \emph{filtration}              & \cite{CoEdHa07,EdLeZo02} \\
                                &                             & \emph{of a complex}            &                          \\
  \hline
  \emph{natural pseudodistance} &  \cite{DoFr04,DoFr07}       &                                &                          \\
  \emph{between size pairs}     &                             &    \rule[3mm]{25mm}{.1pt}      & \rule[3mm]{10mm}{.1pt}   \\
  \hline
  \emph{measuring}              &  \cite{Fr90,Fr96}           & \emph{filtrating}              & \cite{EdHa08,EdLeZo02}   \\
  \emph{function }              &                             & \emph{function}                &                          \\
  \hline
  \emph{(multidimensional)}     &\cite{BiCeFrGiLa08,Fr90,Fr91}& \emph{$0$th rank invariant}    &   \cite{CaZo09}          \\
  \emph{size function}          &                             &                                &                          \\
  \hline
  \emph{size homotopy group}    &  \cite{FrMu99}              &    \rule[1mm]{25mm}{.1pt}      &\rule[1mm]{10mm}{.1pt}    \\
  \hline
  \emph{size functor}           &  \cite{CaFePo01}            &    \rule[1mm]{25mm}{.1pt}      &\rule[1mm]{10mm}{.1pt}    \\
  \hline
  \rule[-1mm]{25mm}{.1pt}       &    \rule[-1mm]{10mm}{.1pt}  & \emph{persistent}              & \cite{EdHa08,EdLeZo02}   \\
                                &                             &  \emph{$k$-th homology group}  &                          \\
  \hline
  \emph{multiplicity}           & \cite{FrLa01,LaFr97}        & \emph{multiplicity of points}  &  \cite{EdHa08,EdLeZo02}  \\
  \emph{of cornerpoints}        &                             & \emph{in persistence diagrams} &                          \\
  \hline
  \emph{formal series}          &                             &                                &                          \\
  \emph{of cornerpoints}        &  \cite{FrLa01,LaFr97}       & \emph{persistence diagrams}    & \cite{EdHa08,EdLeZo02}   \\
  \emph{and cornerlines}        &                             &                                &                          \\
  \hline
  \emph{multidimensional}       &  \cite{BiCeFrGiLa08,FrMu99} & \emph{multidimensional}        & \cite{Ca09,CaZo09}       \\
  \emph{Size Theory}            &                             & \emph{Persistent Homology}     &                          \\
  \hline
\end{tabular}
\vspace{4mm}\caption{Approximative correspondence between some
concepts in Size Theory and Persistent Homology. For each concept
bibliographic references are reported. A line denotes a missing
correspondence.}\label{t1}
\end{table}

As we have already said previously, the objects under study in
Size Theory are the pairs \emph{(topological space, measuring
function)}, called size pairs. The main results of this paper,
stated in Theorem \ref{Teoremone} and Theorem \ref{Teoremone2},
are given under the assumption that the topological space is a
closed $C^1$ Riemannian manifold, while the measuring function is
supposed to be at least continuous. In Persistent Homology the
object of study is usually a simplicial complex $K$, endowed with
a filtration, i.e. a nested sequence of subcomplexes that starts
with the empty complex and ends with the complete complex $K$. The
filtration is usually obtained by a real-valued function defined
at the vertices of $K$ and extended to the simplexes. Each level
$K_c$ in the filtration is obtained by taking just the simplexes
having vertexes at which the function takes a value less than (or
equal to) a parametrical value $c$.

As a matter of fact, Size Theory is more focused on continuous
data (topological spaces or manifolds, endowed with continuous or
$C^k$ functions), while Persistent Homology usually studies
discrete structures (simplicial complexes endowed with piecewise
linear functions) or structures satisfying some finiteness
hypotheses (topological spaces endowed with tame functions). As a
consequence, the results obtained in the two theories are often
expressed and proved in similar but different mathematical
settings. For example, while the fact that the persistent homology
groups are finitely generated is just a trivial consequence of the
assumed hypotheses, the finiteness of size functions requires a
(simple but not trivial) proof. Analogously, while the
localization of discontinuities for the rank of the $0$-th
persistent homology group (i.e. the $0$-th rank invariant) is
usually trivial in the $1$-dimensional setting, this does not hold
for the discontinuities of a size function. This is actually what
happens in this paper, where the measuring functions are not
required to be tame (cf. \cite{CoEdHa07} for a formal definition
of tame function). Obviously, this creates many technical
difficulties, even in the case of $C^1$ measuring functions, since
they are allowed to have an infinite number of critical values.
This kind of problem does not usually appear in literature
regarding Persistent Homology.

Size functions are the most usual tool in Size Theory, while
persistent homology groups constitute the main object of research
in Persistent Homology. Size functions are simply the rank of
persistent $0$-homology groups. On the other hand, the
relationship between persistent homology groups (introduced in
\cite{EdLeZo02}) and size homotopy groups (introduced in
\cite{FrMu99}) is the same that links homology groups and homotopy
groups. For example, the first persistent homology group is the
Abelianization of the first size homotopy group.

Both size functions and persistent homology groups are often
represented by sets of points with multiplicities. The
representation for size functions is called \emph{formal series of
cornerpoints (proper and at infinity)} and was introduced in
\cite{LaFr97}. The correspondent representation for persistent
homology groups is named \emph{persistence diagram} and was
introduced in \cite{EdLeZo02}. The formulas defining the
multiplicities of the considered points are quite analogous.
However, because of the hypotheses usually assumed in Persistent
Homology, persistence diagrams are finite collections of points,
while the formal series used in Size Theory can contain an
infinite number of cornerpoints. The $k$-Triangle Lemma in
\cite{EdLeZo02} is essentially equivalent to the Representation
Theorem recalled in this paper and proved in \cite{FrLa01} (under
slightly different hypotheses).

Formal series representing size functions and persistence diagrams
representing the ranks of persistent homology groups can be
compared by using some matching distances (cf.
\cite{FrLa01,LaFr97} for size functions and \cite{CoEdHa07} for
persistence diagrams). The matching distance used in this paper
has been studied in \cite{dA02,dAFrLa} for size functions and in
\cite{CoEdHa07} for persistent homology groups.

The study of multidimensional measuring functions has started in
\cite{FrMu99} for Size Theory and in \cite{CaZo09} for Persistent
Homology.




\end{document}